\documentclass[journal,twoside,web]{ieeecolor}
\usepackage{generic}
\usepackage{cite}
\usepackage{amsmath,amssymb,amsfonts}
\usepackage{algorithmic}
\usepackage{graphicx}
\usepackage{algorithm,algorithmic}
\usepackage[colorlinks=true, linkcolor=blue, urlcolor=blue, citecolor=blue]{hyperref}
\usepackage{textcomp}

\usepackage{mathtools}

\usepackage{tikz}

\usepackage{bm}
\usepackage{here}
\usepackage{mathrsfs}

\usepackage{stmaryrd}
\usepackage{xcolor}
\definecolor{olive}{rgb}{0.5, 0.5, 0.0} 
\definecolor{mistygray}{rgb}{0.93, 0.93, 0.95}
\definecolor{tangerine}{rgb}{0.949, 0.522, 0.0}
\definecolor{plum}{rgb}{0.60, 0.32, 0.60}

\usepackage[subrefformat=parens]{subcaption}

\newtheorem{proposition}{\bf Proposition}
\newtheorem{definition}{\bf Definition}
\newtheorem{lemma}{\bf Lemma}
\newtheorem{theorem}{\bf Theorem}

\newtheorem{rmk}{\bf Remark}

\newtheorem{assumption}{\bf Assumption}

\newcommand{\calI}{{\mathcal I}}    
    
\newcommand{\calK}{{\mathcal K}}    
\newcommand{\calL}{{\mathcal L}}

\newcommand{\calP}{{\mathcal P}}

\newcommand{\calS}{{\mathcal S}}

\newcommand{\bbC}{{\mathbb C}}

\newcommand{\bbR}{{\mathbb R}}

\newcommand{\bbZ}{{\mathbb Z}}

\newcommand{\rmd}{{\rm d}}

\newcommand{\iu}{{\rm i}}   
\newcommand{\Ee}{{\rm e}}   

\newcommand{\kl}[2]{D_{\rm KL} ({#1}\|{#2})}
\newcommand{\FI}[2]{\calI \left({#1}\middle\|{#2}\right)}

\newcommand{\what}[1]{\widehat{#1}}

\newcommand{\FF}{{\rm FF}}
\newcommand{\FB}{{\rm FB}}

\newcommand{\modi}{{\rm FF}}

\newcommand{\lw}{{\underline{u}}}
\newcommand{\up}{{\overline{u}}}

\newcommand{\osc}{{\rm osc}}

\newcommand{\wtilde}[1]{\widetilde{#1}}

\newcommand{\rev}[1]{{\color{black}{#1}}} 
\newcommand{\ito}[1]{{\color{black}{#1}}} 
\newcommand{\hi}[1]{{\color{black}{#1}}} 
\newcommand{\revv}[1]{{\color{black}{#1}}} 

\def\BibTeX{{\rm B\kern-.05em{\sc i\kern-.025em b}\kern-.08em
    T\kern-.1667em\lower.7ex\hbox{E}\kern-.125emX}}
\begin{document}
\title{Distribution Control of Stochastic Oscillators via Periodic Feedforward and Population-Level Feedback}
\author{Kaito Ito, Haruhiro Kume, and Hideaki Ishii
\thanks{This work was supported in part by JSPS KAKENHI Grant Numbers JP23K19117, JP24K17297, JP22H01508, JP23K22778 and JST, ACT-X Grant Number JPMJAX2102. }
\thanks{Kaito Ito and Hideaki Ishii are with the Department of Information Physics and Computing, The University of Tokyo, Tokyo 113-8654, Japan (e-mail: kaito@g.ecc.u-tokyo.ac.jp, hideaki\_ishii@ipc.i.u-tokyo.ac.jp). }
\thanks{Haruhiro Kume is with the Department of Systems and Control Engineering, Tokyo Institute of Technology, Yokohama 226-8502, Japan (e-mail: \hi{hal3000.hal98.88@gmail.com}).}
}

\maketitle

\begin{abstract}
We address the problem of steering the phase distribution of oscillators all receiving the same control input to a given target distribution. In a large population limit, the distribution of oscillators can be described by a probability density.
Thus, we formulate the problem as \revv{distribution control over a Fokker--Planck equation}.
In particular, we consider the case where oscillators are \hi{subject} to stochastic noise, for which the theoretical understanding is still lacking.
First, we characterize the \revv{asymptotic} reachability of the phase distribution under periodic feedforward control via the Fourier coefficients of the target density and the phase sensitivity function of oscillators. This enables us to design a periodic input that makes the stationary distribution of oscillators closest to the target by solving a convex optimization problem.
Next, we devise a \revv{distribution controller} combining periodic and \revv{population-level} feedback control, where the feedback component is designed to accelerate the convergence of the distribution of oscillators.
We exhibit some convergence results for the proposed method, including a result that holds even under measurement errors in the phase distribution.
The effectiveness of the proposed method is demonstrated by a numerical example.
\end{abstract}


\section{Introduction}
\label{sec:intro}

\IEEEPARstart{P}{opulations} of oscillators provide a powerful framework for modeling a wide range of collective phenomena, including neuronal ensembles, pedestrian crowds, and firefly swarms.
For example\hi{,} in neuroscience, several diseases---such as Parkinson's disease, Alzheimer's disease, and sleep disorders---can be explained as pathological synchronizations of neural oscillator populations~\cite{Hanahan,Zhivotovsky,Kane}.
Treating these diseases with external stimuli via medicine or electrical signals can be modeled as controlling the population of oscillators by an external input so that it behaves normally\hi{\cite{Franci2012,Wang2012,Kokubo2024}}.
One of the difficulties of this control problem is that oscillators receive the same input because it is impossible to apply different inputs to each cell.

\revv{Control problems for oscillator populations driven by a common
input have been studied in close connection with distribution control and ensemble
control~\cite{Tass,Brown2004,Wilson,Monga2018,Kuritz,Kato}.}
In \cite{Monga2018}, control of a deterministic oscillator population is
formulated as a distribution control problem, where the distribution of oscillators is
described as a probability density function.
Then, the authors proposed a \revv{population-level} feedback control law that decreases the $ L^2 $~distance between the density of oscillators and a prescribed target density. \revv{Here, ``population-level'' means that the current density of oscillators is used to determine the control input.}
Moreover, for the same control method, \hi{it was shown in} \cite{Kuritz} that the Fourier coefficients of the so-called phase sensitivity function of oscillators play a crucial role in the convergence of the distribution of oscillators to the target.

In the above deterministic \hi{cases}, once the distribution of oscillators is transferred to a desired distribution, then the distribution keeps the desired shape without any control.
However, when taking into account random fluctuations of oscillators, which commonly occur in realistic settings, the situation is quite different.
Even if the distribution is steered close to a target distribution, the stochastic oscillators mix to the uniform distribution in the absence of continued control.
Thus, we need to use an appropriate control input continually to keep the stochastic oscillators close to the target.
For stochastic oscillator populations under common periodic
inputs, the work \cite{Kato} considered the approximated dynamics of
their distribution obtained via the so-called averaging method~\cite{Pavliotis2008multiscale}.
Then, the authors formulated an optimization problem which yields a periodic feedforward input that makes the stationary distribution of the approximated dynamics closest to the desired distribution.
However, the transient response of oscillators cannot be taken into account \hi{there} because the optimization is only based on the stationary distribution. In addition, the formulated optimization problem is nonconvex, which poses a challenge in finding a globally optimal solution.
Moreover, there is a lack of theoretical understanding regarding how accurately the resulting stationary distribution can approximate the target distribution.

\textit{Contributions:}
In this paper, we provide \hi{a} theoretical understanding of the \revv{distribution} control of stochastic oscillators concerning reachability and convergence properties.
Building on this analysis, we propose a control law that combines population-level feedback and periodic feedforward control, aiming to address the limitations of each approach.
More precisely, our contributions are as follows:

	1) We analyze the \revv{asymptotic} reachability of the distribution of stochastic oscillators. First, we investigate the approximated phase equation of oscillators derived using the averaging method. Then, we show that under periodic feedforward control, the reachability of stochastic oscillators to a target distribution is characterized by the Fourier coefficients of their phase sensitivity function and the target (Theorems~\ref{thm:exact_reachability} \hi{and} \ref{thm:appro_reachability}). By using Theorem~\ref{thm:appro_reachability}, the minimization of the difference between the stationary distribution of oscillators and the target distribution can be formulated as a convex optimization problem, which determines the Fourier coefficients of a periodic input. Moreover, to bridge the gap between the original and averaged phase equations, we bound the Kullback--Leibler (KL) divergence between the distribution of oscillators and the stationary distribution of the averaged equation (Theorem~\ref{thm:bound_average_original}). Consequently, we can efficiently compute an optimal periodic input that steers the distribution of oscillators close to the target, where the closeness is guaranteed by the derived upper bound.
	
	2) Next, we establish convergence properties for the \revv{distribution} control of stochastic oscillators. 
	We devise a control law combining the designed periodic control and a feedback controller, under which the convergence of oscillators is guaranteed (\hi{Theorem}~\ref{thm:proposed}).
	\revv{This result shows that even under the designed periodic input alone, the distribution of oscillators converges exponentially to the same distribution. In addition, the proposed feedback control adds an extra nonpositive dissipation term to a Lyapunov functional based on the KL divergence, and its acceleration effect on the transient response is demonstrated numerically.}
	Moreover, we modify the proposed controller so that even under measurement errors of the distribution of oscillators, the convergence is guaranteed (Proposition~\ref{prop:measurement_err}), which improves the applicability of our method. 

A preliminary version~\cite{Ito2023ensemble} of this work presented a special case of \hi{Theorem}~\ref{thm:proposed}, where the diffusion coefficients of oscillators are constant. 
The \hi{current} paper provides further results and discussions.

\textit{Related work:}
Population control with a common input has been extensively studied in the literature on ensemble control\hi{\cite{Li2016,Chen2022sparse,Chen2019controllability,Helmke2014uniform,Mandal2023control,Vu2024data}}.
In \cite{Li}, an ensemble of
systems is represented by a parameterized family of control systems driven by the same control, and the associated controllability condition was derived. \revv{This parameter-indexed formulation is different from the setting of the present paper, where the oscillators are not indexed by a heterogeneity parameter and the control objective is to steer the population density.}
In \cite{Azuma2013}, multi-agent coordination by broadcasting a common signal to all agents was considered, and it was revealed that randomness of a local controller of each agent is essential for achieving given motion-coordination tasks.
As a dual problem of ensemble control, \hi{the works} \cite{Zeng,ZengIshii} studied state estimation problems of ensembles that are expressed by probability distributions and characterized the observability of ensembles.

Despite extensive literature on population control for general systems under common inputs, relatively few studies have addressed the distribution control of an oscillator population driven by such inputs.
In \cite{Wilson2014}, desynchronization of a neural population is formulated as distribution control of oscillators so that the peak of their distribution is decreased using \hi{limited} energy. 
Our problem setting is more challenging because the shape of the distribution should be considered, not only the peak.

Based on the control method in \cite{Monga2018} and Fourier analysis, \hi{the work} \cite{Monga} decomposed a partial differential equation, which governs the evolution of the distribution of oscillators, into ordinary differential equations governing the evolution of the Fourier coefficients of the distribution. Consequently, the distribution control reduces to the control of the Fourier coefficients. For stochastic oscillators, they also proposed the population-level feedback controller, which cancels the diffusion term of the partial differential equation so that the $ L^2 $~distance between the distribution of oscillators and the target decreases to zero. However, the resulting controller is not well-defined in general, and experimentally, the cancellation term induces undesirable behavior of the distribution as we will see in Subsection~\ref{subsec:previous_feedback} and Section~\ref{sec:example}.

\textit{Organization:}
This paper is organized as follows. In Section~\ref{sec:previous}, we briefly introduce an oscillator model and the issues of the existing distribution control methods. 
In Section~\ref{sec:reachability}, we analyze the reachability of the distribution of stochastic oscillators and propose an efficient design method of a periodic input steering the phase distribution close to the target.
In Section~\ref{sec:controller}, we devise a distribution control method combining the periodic feedforward control and a feedback controller and then show its convergence properties.
In Section~\ref{sec:example}, a numerical example illustrates the effectiveness of the proposed method.
Finally, Section~\ref{sec:conclusion} concludes this paper.

\textit{Notations:}
Let $ \bbR $, $ \bbC $, and $ \bbZ $ denote the sets of real numbers, complex numbers, and integers, respectively.
The unit circle is denoted by $ \calS^1 $, and we identify \(\calS^1\) with \(\mathbb{R}/2\pi\mathbb{Z}\).
The $ L^p $~norm for $ p \ge 1 $ and the $ L^\infty $~norm of a function $ f $ on $ \calS^1 $ are defined by $ \|f\|_p := (\int_{\calS^1} |f(\theta)|^p \rmd \theta)^{1/p} $ and $ \|f\|_\infty := \inf \{ M \ge 0 : |f| \le M ~ \text{almost everywhere} \} $, respectively. 
Denote by $ L^2(S) $ the space of square-integrable functions on $ S $. For an interval $ (a,b) $, we also write $ L^2 ((a,b)) $ as $ L^2 (a,b) $. Denote by $ \ell^2 (\bbZ) $ (resp. $ \ell^1 (\bbZ) $) the space of square-summable (resp. absolutely summable) sequences indexed $ \bbZ $.
The set of all $ k $-times continuously differentiable functions on $ \calS^1 $ is denoted by $ C^k(\calS^1) $.
Denote $ \partial f(\theta)/\partial \theta $ and $ \partial^2 f(\theta)/\partial \theta^2 $ by $ \partial_\theta f(\theta) $ \hi{and} $\partial_\theta^2 f(\theta) $, respectively. When no confusion can arise, we will omit arguments of functions and the domain of integration as $ \int f \rmd \theta = \int_{\calS^1} f(\theta)\rmd \theta $.
The imaginary unit is denoted by $ \iu $.
Let $ \calP(\calS^1) $ be the set of all density functions on $ \calS^1 $.
The KL divergence between probability densities $ \rho_1 $ and $ \rho_2 $ is denoted by $ \kl{\rho_1}{\rho_2} := \int_{\calS^1} \rho_1(\theta) \log (\rho_1(\theta) / \rho_2 (\theta)) \rmd \theta $.
Define
\begin{equation}
    \FI{\rho_1}{\rho_2} := \int_{\calS^1} \rho_1 (\theta) \left( \partial_\theta \log \frac{\rho_1(\theta)}{\rho_2 (\theta)} \right)^2 \rmd \theta , \nonumber
\end{equation}
which is called the relative Fisher information (or Fisher divergence) of $ \rho_1 $ with respect to $ \rho_2 $~\cite{Villani}.
The KL divergence and the relative Fisher information are nonnegative functionals and take the value $ 0 $ if and only if $ \rho_1 = \rho_2 $.

\section{Oscillator Model in a Large Population Limit and Existing Distribution Control Methods}\label{sec:previous}
In this section, we introduce an oscillator model in a large population limit and the existing control methods of the distribution of oscillators~\cite{Monga2018,Kuritz,Monga,Kato}.
\subsection{Oscillator Model in a Large Population Limit} 
First, we consider identical, uncoupled $ N $ oscillators following the \hi{stochastic} phase model on the unit circle $ \calS^1 $ obtained by phase reduction~\cite{Winfree,Kuramoto,Monga}:
\begin{align}
\rmd \theta_i (t) = (\omega + Z(\theta_i (t)) u(t) )\rmd t + \sqrt{2D} Z_w(\theta_i(t)) \rmd W_i (t) , \label{eq:phasemodel}
\end{align}
where \(i\in\{1,\ldots,N\}\), $ \theta_i(t) \in \calS^1 $ denotes the phase of the $ i $th oscillator, \(\omega\in\mathbb{R}\setminus\{0\}\) denotes the natural frequency, \(u(t)\in\mathbb{R}\) denotes a common external control input, \(D>0\) denotes the noise intensity, and \(\{W_i\}\) are independent standard Wiener processes.
The functions $ Z, Z_w : \calS^1 \rightarrow \bbR $ are called the phase sensitivity functions (or phase response curves), which describe the linear response of the phase to an input and a noise, and it is assumed that $ Z, Z_w \in C^2 (\calS^1) $.
\revv{The above equation is understood as an Ito diffusion on \(\calS^1\) in the following sense. When \eqref{eq:phasemodel} is regarded as a diffusion on $ \bbR $ with the \(2\pi\)-periodic lifts of $ Z $ and $ Z_w $ to $ \bbR $ and a real-valued Wiener process, its solution is denoted by $ X_i (t) \in \bbR $. Then, the solution to \eqref{eq:phasemodel} on $ \calS^1 $ is defined as $ \theta_i(t) := X_i (t) \bmod 2\pi  $.} 

We emphasize that all the oscillators are driven by the same input \(u\). Without the input and the noise, i.e., \(u\equiv0\) and \(D=0\), the oscillators rotate with the constant angular velocity \(\omega\). The model~\eqref{eq:phasemodel} appears, for example, in neuroscience~\cite{Tass,Kuritz}, where it is difficult to apply different control inputs to each oscillator.

\revv{The main focus of this paper is the macroscopic control of the oscillator population. Suppose that the initial phases $\{\theta_i(0)\}_{i=1}^N$ are independent and identically distributed (i.i.d.) according to a probability density function $\rho_0$ and that $ u $ is deterministic. Since the oscillators are uncoupled and driven by independent Wiener processes, the phase processes $\{\theta_i(t)\}_{i=1}^N$ are i.i.d. for each $t\ge 0$. Therefore, by the standard law of large numbers for empirical measures~\cite{van2000asymptotic}, the empirical distribution constructed from $\{\theta_i(t)\}_{i=1}^N$ converges, as \(N\to\infty\), to a deterministic probability measure; see Fig.~\ref{fig:circ_histo_pdf}. Let us assume that the resulting measure admits a density function $ \rho : [0,\infty) \times \calS^1 \rightarrow \bbR $. This density represents the macroscopic phase distribution of the oscillator population and satisfies the Fokker--Planck equation associated with the Ito diffusion~\eqref{eq:phasemodel}:}
\begin{align}
    \partial_t\rho(t,\theta) &= -\partial_\theta\left[(\omega + Z(\theta)u(t))\rho(t,\theta)\right] + D \partial_\theta^2[ Z_w^2 (\theta) \rho(t,\theta)] \nonumber\\
	&=: \calL_{u(t)}\rho(t,\theta) , \label{eq:rho} \\
	\rho(0,\theta) &= \rho_0(\theta) . \nonumber
\end{align}
Here, $ \rho(t,\theta)\rmd \theta $ can be interpreted as the fraction of oscillators whose phases lie within the interval $ [\theta ,\theta + \rmd \theta) $ at time $ t $.
\rev{The situation where all the oscillators receive a common input means that $ u $ is not allowed to depend on $ \theta $.}
Throughout this paper, we assume the existence of a unique solution to \eqref{eq:rho}.

Then, our goal is to steer the distribution of oscillators $ \rho $ to a given target distribution.
Specifically, the target distribution $ \rho_f $ is given by
\begin{equation}
	\begin{aligned}
		\partial_t\rho_f(t,\theta) &= -\omega\partial_\theta\rho_f(t,\theta),  \\
		\rho_f(0,\theta) &= \rho_{f,0}(\theta) ,  \label{eq:rhof}
	\end{aligned}
\end{equation}
where the solution $ \rho_f $ rotates on $ \calS^1 $ at the angular velocity $ \omega $ maintaining the shape of the initial density $ \rho_{f,0} $. Hence, we sometimes refer to $ \rho_{f,0} $ as the target as well.

\revv{
\emph{Overview of Main Assumptions:}
We highlight several key assumptions used in the main results. These assumptions are specific to the corresponding results, and the complete assumptions are stated in the theorem statements. Theorems~\ref{thm:exact_reachability} and \ref{thm:appro_reachability} (asymptotic reachability under averaging) assume that a periodic input $ u_\FF $ belongs to $ L^2 (0,2\pi/\omega) $.
For exact reachability in Theorem~\ref{thm:exact_reachability}, the Fourier coefficient of $ Z $ corresponding to each nonzero Fourier coefficient of $ \partial_\theta \log \rho_{f,0} $ must be nonzero.
Theorems~\ref{thm:bound_average_original} and \ref{thm:proposed} both assume uniform boundedness of the phase density generated by the periodic feedforward input.
Theorem~\ref{thm:bound_average_original} (averaging error) additionally uses Assumption~\ref{ass:nondegenerate} ($ Z_w \equiv 1 $), whereas Theorem~\ref{thm:proposed} (exponential convergence) assumes the non-degeneracy condition $ Z_w^2 > 0 $.

}

\subsection{Existing \revv{Population-Level} Feedback Control Methods and Their Issues}\label{subsec:previous_feedback}
For the deterministic case (i.e., with $ D = 0 $), \cite{Monga2018} designed a population-level feedback control law that decreases the $ L^2 $~distance between $ \rho $ and the target $ \rho_f $. 
\revv{This controller assumes that the current distribution $ \rho(t,\cdot) $ is available to determine $ u(t) $.
Such a population-level feedback formulation has also appeared in the existing studies on phase-distribution control~\cite{Kuritz,Monga}, robotic swarm control~\cite{Zheng2021transporting}, and control of general Fokker--Planck equations~\cite{Breiten2018control}.}
 In addition, it was revealed in \cite{Kuritz} that if all Fourier coefficients of $ Z $ are nonzero, then $ \rho $ converges to $ \rho_f $ under the control law proposed in \cite{Monga2018}.

 \begin{figure}[tb]
  \centering
  \includegraphics[width=1.03\linewidth]{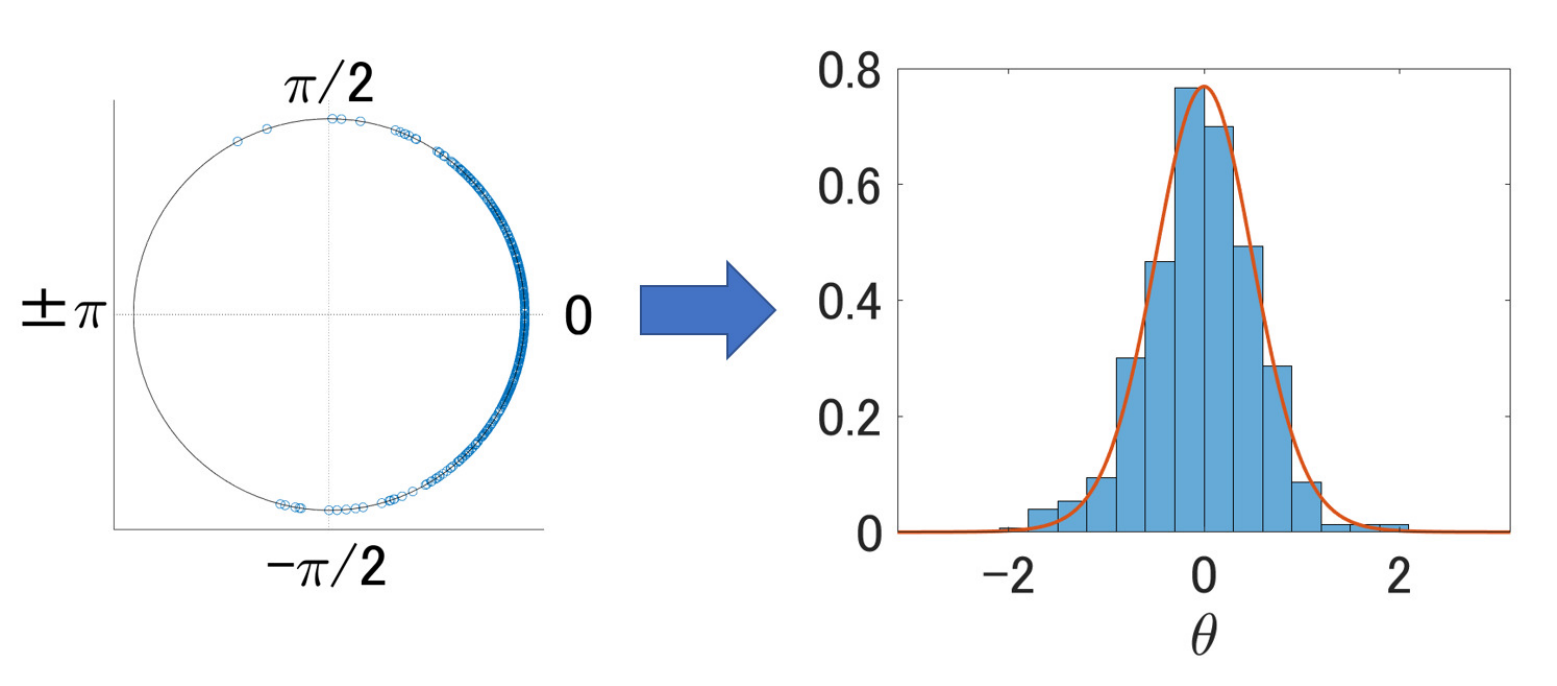}
  \caption{Distribution of infinitely many oscillators can be represented as a density function (\hi{red line}). Each oscillator is regarded as an independent sample drawn from the density. As the number of oscillators tends to infinity, their empirical distribution (blue histogram) converges to the density.}
  \label{fig:circ_histo_pdf}
\end{figure}

However, when $ D > 0 $, the convergence to \hi{the target} $ \rho_f $ cannot be achieved in general.
This can be seen \rev{by considering the evolution of the error $ \Delta \coloneqq \rho-\rho_f $:}
\begin{align}
  &\partial_t \Delta = -\omega\partial_\theta\Delta - \partial_\theta[Z(\theta)\rho(t,\theta)]u(t) + D \partial^2_\theta \left[Z_w^2 (\theta) \rho(t,\theta) \right]. \label{eq:delta}
\end{align}
\rev{That is, in order to make $ \Delta = 0 $\hi{,} a stationary solution to \eqref{eq:delta}, the control input $ u $ must satisfy for any $ \theta \in \calS^1 $,
\[
	\lim_{t\rightarrow \infty} \left( \partial_\theta [Z(\theta) \rho_f(t,\theta)] u(t) - D \partial_\theta^2 \left[ Z_w^2(\theta) \rho_f (t,\theta) \right] \right) = 0  .
\]
However, this condition cannot be fulfilled by $ u $ which does not depend on $ \theta $ except for special cases, such as when $ Z_w $ is a constant function and the target distribution is uniform $ \rho_{f} \equiv 1/(2\pi) $.
Even} \ito{when $ D > 0 $}, the control law in \cite{Monga2018,Kuritz} decreases the $ L^2 $~distance between $ \rho $ and $ \rho_f $ to some extent.
Indeed, the time derivative of the squared $ L^2 $~norm $ V(\Delta(t,\cdot)) := \frac{1}{2} \|\Delta(t,\cdot)\|_2^2 $ along \eqref{eq:delta} is given by
\begin{align}
	\frac{\rmd V(\Delta)}{\rmd t} &= \int_{\calS^1} \Delta \left( -\omega\partial_\theta\Delta - \partial_\theta [Z\rho]u + D \partial_\theta^2 [Z_w^2 \rho] \right)\rmd \theta \nonumber\\
	&= \left(\int Z\rho \partial_\theta \Delta \rmd \theta\right) u + D \int \Delta  \partial_\theta^2 [Z_w^2 \rho] \rmd \theta \rev{,} \label{eq:L2derivative}
\end{align}
where the argument $ t $ is omitted.
Therefore, when $ D = 0 $, the control law used in \cite{Monga2018,Kuritz}:
\begin{equation}\label{eq:previous}
	u(t) = - k \int_{\calS^1} Z(\theta)\rho(t,\theta) \partial_\theta \Delta(t,\theta) \rmd \theta, \quad k > 0,
\end{equation}
ensures $ \rmd V/\rmd t \le 0 $ and monotonically decreases $ V $. With the diffusion term (i.e., $ D > 0 $), while $ D \int \Delta \partial_\theta^2 [Z_w^2 \rho] \rmd \theta \le k(\int Z\rho \partial_\theta \Delta\rmd \theta)^2 $ holds, the control law \rev{\eqref{eq:previous}} decreases $ V $.
The resulting deviation of $ \rho $ from $ \rho_f $ may be small.
However, it is difficult to estimate the deviation \hi{analytically}. 

\hi{In this paper, we approach this problem in a more rigorous manner.} We first provide a surrogate target distribution fluctuating around $ \rho_f $, and then design a population-level feedback control law steering oscillators to the surrogate target exactly.
Its advantage is that if the deviation between the surrogate and original targets is guaranteed to be small, the resulting stationary deviation between $ \rho $ and the original target $ \rho_f $ is also small.
In the next section, we explain that such a surrogate target can be designed by periodic feedforward control.

We also mention the existing study~\cite{Monga}, which proposed the control law:
\begin{equation}\label{eq:proposed_cancel}
	u(t) = - k \int  Z\rho \partial_\theta \Delta \rmd \theta - \frac{D\int \Delta  \partial_\theta^2 [Z_w^2 \rho] \rmd \theta  }{\int Z\rho \partial_\theta \Delta \rmd \theta}\rev{,} \quad k > 0 ,
\end{equation}
whose second term cancels the second term of \eqref{eq:L2derivative}. The authors asserted that this control law decreases the $ L^2 $~distance between $ \rho $ and $ \rho_f $ until $ \rho $ becomes equal to $ \rho_f $. However, this is obviously not the case in general since there is \hi{no} input \rev{under which $ \rho $ converges to $ \rho_f $ as already observed in \eqref{eq:delta}. The $ L^2 $~distance converges to a strictly positive constant under the assumption $ \int Z\rho \partial_\theta \Delta \rmd \theta \neq 0 $ for all $ t $, which ensures the well-definedness of the control law \eqref{eq:proposed_cancel}. However, this assumption is not satisfied in general as we will numerically observe through an example in Section~\ref{sec:example}.}

\subsection{Existing Feedforward Control Method and Its Issues}\label{sec:modified_target}
In the previous subsection, we observed that the distribution of stochastic oscillators cannot be stabilized to the target distribution given by \eqref{eq:rhof}.
Instead of exact stabilization, \hi{in this subsection, we explain that} maintaining the distribution $ \rho $ near the target $ \rho_f $ can be effectively achieved by continuously applying a properly designed periodic input.

\hi{Consider} a distribution of oscillators $ \rho_{\modi} $ driven by a feedforward input $ u_{\FF} $. Then, $ \rho_{\modi} $ evolves as follows:
\begin{align}
  \partial_t\rho_\FF (t,\theta) &= \calL_{u_\FF(t)}\rho_\modi (t,\theta) \nonumber\\
	&= -\partial_\theta \! \left[(\omega + Z(\theta)u_{\rm FF}(t))\rho_\modi (t,\theta)\right] \nonumber\\
  &\quad + D \partial_\theta^2 \! \left[ Z_w^2 (\theta) \rho_\modi (t,\theta) \right] \label{eq:rho1},\\
	\rho_\modi (0,\theta) &= \rho_{\modi,0} (\theta)  . \nonumber
\end{align}
For notational convenience in the subsequent discussion, we use the notation $ \rho_{\modi} $ instead of $ \rho $.
\hi{In} what follows, focusing on $ 2\pi/\omega $-periodic inputs, we explain how to design $ u_\FF $ that makes $ \rho_\modi $ close to $ \rho_{f} $ based on the work~\cite{Kato}.

\hi{First}, we perform the change of variables $ \wtilde{\rho}_\modi (t,\theta) := \rho_{\modi} (t,\theta +\omega t) $ in \eqref{eq:rho1}, which corresponds to the coordinate transformation $\theta\mapsto\theta-\omega t $ \rev{with} $ t\mapsto t$. Then, the resulting density $ \wtilde{\rho}_\modi $ satisfies
\begin{align}
  \partial_t \wtilde{\rho}_\modi (t,\theta) &= -\partial_\theta\left[(Z(\theta+\omega t)u_\FF (t))\wtilde{\rho}_\modi (t,\theta)\right] \nonumber\\
	&\quad + D \partial_\theta^2 \left[ Z_w^2 (\theta + \omega t) \wtilde{\rho}_\modi (t,\theta) \right]  \nonumber \\
    &=: \wtilde{\calL}_{\revv{t},u_\FF} \wtilde{\rho}_\modi (t,\theta) . \label{eq:rho1-rot}
\end{align}
In the rotating frame $\theta\mapsto\theta-\omega t $, the target distribution is equal to its initial density $ \rho_{f,0} $.
\revv{Next, following the averaging method~\cite{Pavliotis2008multiscale}, we introduce an averaged dynamics for
\eqref{eq:rho1-rot}. 
We define the
period-averaged forward operator by
\begin{align}
\bar{\mathcal L}f(\theta)
&:=
\frac{\omega}{2\pi}
\int_0^{2\pi/\omega}
\widetilde{\mathcal L}_{t,u_\FF} f(\theta) \rmd t \nonumber
\\
&= -
\partial_\theta [\Gamma(\theta) f(\theta)]
+
B^2\partial_\theta^2 f(\theta),
\end{align}
where
\begin{align}
\Gamma(\theta)
&:=
\frac{\omega}{2\pi}
\int_0^{2\pi/\omega}
Z(\theta+\omega t)u_{\mathrm{FF}}(t) \rmd t, \label{eq:gamma_def} \\
\qquad
B^2
&:= D \frac{\omega}{2\pi} \int_0^{2\pi/\omega}   Z_w^2 (\theta + \omega t) \rmd t =
\frac{D}{2\pi}\int_0^{2\pi}Z_w^2(\vartheta)d\vartheta . \nonumber
\end{align}
Then, the averaged Fokker--Planck equation is defined as the
forward equation associated with \(\bar{\mathcal L}\):
\begin{align}
\partial_t\bar\rho_{\mathrm{FF}}(t,\theta) &= \bar{\calL} \bar{\rho}_\FF (t,\theta) \nonumber\\
&=
-\partial_\theta[\Gamma(\theta)\bar\rho_{\mathrm{FF}}(t,\theta)]
+
B^2\partial_\theta^2\bar\rho_{\mathrm{FF}}(t,\theta). \label{eq:rho1-mean}
\end{align}
We utilize the solution $ \bar{\rho}_{\FF} $ to the above equation as an approximation of $ \wtilde{\rho}_{\FF} $ or $ \rho_{\FF} $.
}It is known that when $ u_\FF $ is small, the averaged equation \eqref{eq:rho1-mean} approximates \eqref{eq:rho1-rot} well~\revv{\cite[Chapter~9.5]{Hoppensteadt}}. The Fokker--Planck equation \eqref{eq:rho1-mean} is time-invariant and has the stationary distribution
\begin{align}
  \bar{\rho}_{\rm st}(\theta) := \frac{1}{C}\int_{\theta}^{\theta+2\pi}\exp\left(-\frac{\int_{\theta}^{\psi}\Gamma(\phi)\rmd\phi}{B^2}\right) \rmd\psi, \label{eq:stationary_averaged}
\end{align}
where $ C > 0$ is the normalizing constant~\cite[Chapter~9.2.2]{Pikovsky2001}.
Then, we aim to find $ u_\FF $ that makes $ \bar{\rho}_{\rm st} $ close to the desired shape $ \rho_{f,0} $. 

\hi{To this end, the work} \cite{Kato} considered the following optimization problem:
\begin{align}
  \underset{u_\FF}{\text{minimize}} ~~ &\|\bar{\rho}_{\rm st} - \rho_{f,0}\|_2^2  \label{prob:fok-opt}\\
  \text{subject~to} ~~ &\|u_\FF (\cdot/\omega)\|_2^2 = E, \label{eq:fok-opt_constraint}
\end{align}
where the constraint \eqref{eq:fok-opt_constraint} fixes the input energy of $ u_\FF $ to $ E > 0 $.
\rev{This problem} can then be solved numerically by approximating it to a finite dimensional problem, for example\hi{,} by a finite difference method and truncating Fourier series.

We would like to highlight several issues of \cite{Kato} from both numerical
and theoretical perspectives: (i) Problem \eqref{prob:fok-opt} is nonconvex, making it challenging to find an optimal solution efficiently. (ii) There is no theoretical analysis clarifying how well the stationary distribution $ \bar{\rho}_{\rm st} $ approximates the target $ \rho_{f,0} $. (iii) The convergence of $ \rho $ under $ u = u_\FF $ is not established analytically. (iv) As we will see in Section~\ref{sec:example}, convergence based on their feedforward
control method may be slow. (v) Although it is known that under a weak input of $ O(\varepsilon) $, the averaged equation is correct up to $ O(\varepsilon) $, a specific upper bound of the difference between $ \wtilde{\rho}_\modi $ and the stationary distribution $ \bar{\rho}_{\rm st} $ of the averaged equation is not known to the best of our knowledge. 

In Section~\ref{sec:reachability}, we address (i), (ii), and (v) analytically.
\revv{Then, in Section~\ref{sec:controller}, to resolve (iii) and (iv), we devise a feedback controller that, when combined with a periodic input $u_{\rm FF}$, ensures the convergence of $\rho$ to the (typically periodic) stationary distribution of $\rho_\modi $.\footnote{In this paper, we distinguish between two types of stationary distributions: the periodic stationary distribution, which refers to the time-periodic steady-state distribution of $ \rho_\modi $ or $ \wtilde{\rho}_\modi $ that typically emerges under periodic inputs, and the stationary distribution $ \bar{\rho}_{\rm st} $, which is time-invariant in the usual sense.} Although this convergence is already guaranteed under the periodic feedforward input alone, this feedback control adds an extra nonpositive dissipation term to the KL Lyapunov functional. Its acceleration effect on the transient response is demonstrated numerically in Section~\ref{sec:example}.}
Unlike the convergence of $ \rho $ to the target $ \rho_f $, which does not hold as we observed in Subsection~\ref{subsec:previous_feedback}, the convergence of $ \rho $ to $ \rho_{\modi} $ can be proved by exploiting the contraction property of Fokker--Planck equations. 
Moreover, this convergence result suggests that the periodic stationary distribution of $ \rho_{\modi} $ can \hi{serve} as a surrogate target, which is close to the original target $ \rho_f $.

\begin{figure}[tb]
  \centering
  \includegraphics[scale=0.37]{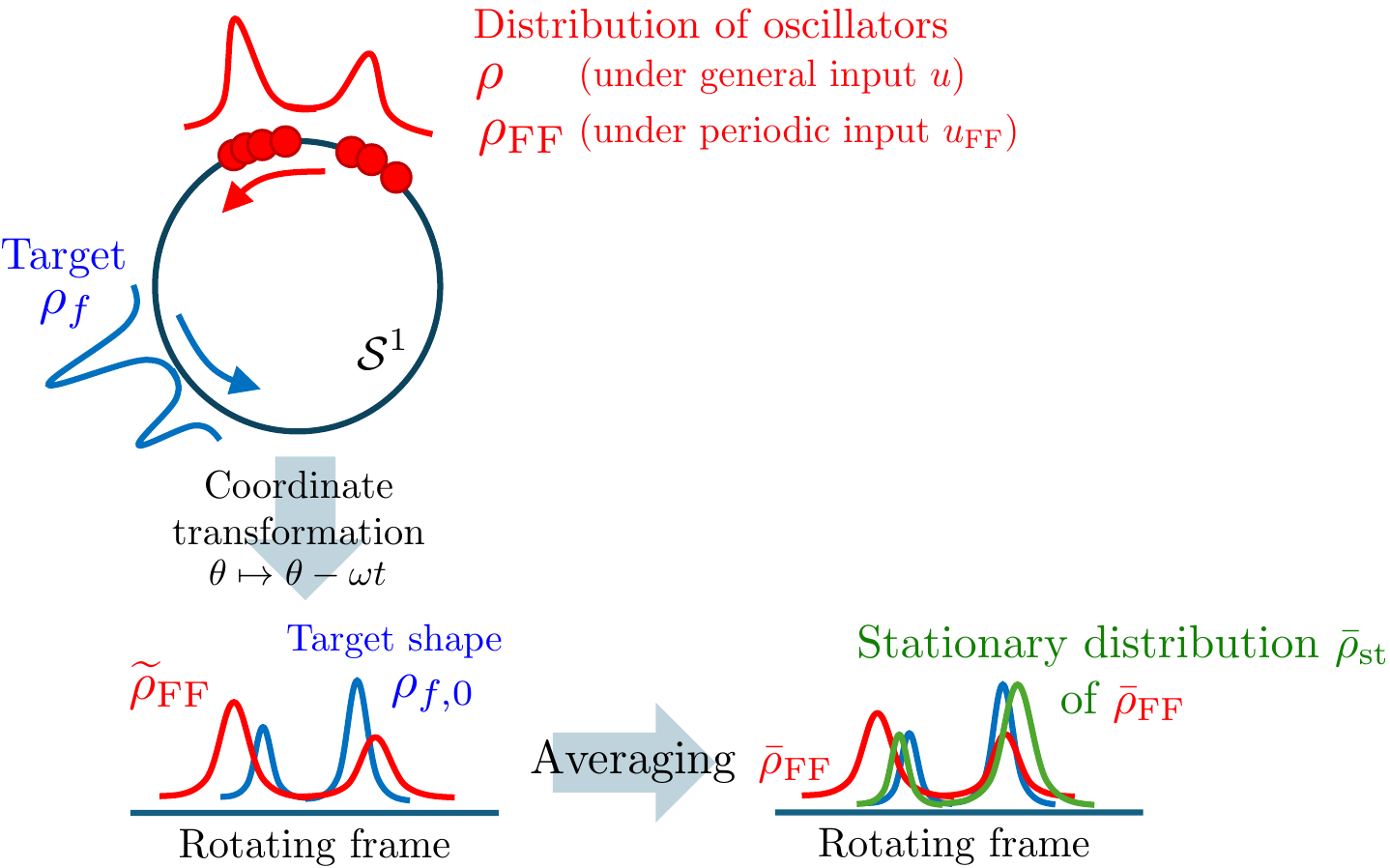}
	   \put(-156,40){\footnotesize \colorbox{mistygray}{\shortstack{Error bound\\of averaging \\ (Theorem~\ref{thm:bound_average_original})}}}
		 \put(-85,55){\footnotesize \colorbox{mistygray}{\shortstack{Error bound\\between \textcolor{olive}{$ \bar{\rho}_{\rm st} $} and \textcolor{blue}{$ \rho_{f,0} $}\\ (Theorem~\ref{thm:appro_reachability})}}}
  \caption{Relationship \hi{among} $ \rho $, $ \rho_f $, $ \rho_\modi $, $ \wtilde{\rho}_\modi $, $ \bar{\rho}_\modi $, $ \bar{\rho}_{\rm st} $, and \hi{the} main results (Theorems~\ref{thm:appro_reachability} \hi{and} \ref{thm:bound_average_original}).}
  \label{fig:summary}
\end{figure}

Lastly, the relationship \hi{among} $ \rho $, $ \rho_f $, $ \rho_\modi $, $ \wtilde{\rho}_\modi $, $ \bar{\rho}_\modi $, and $ \bar{\rho}_{\rm st} $ is illustrated in Fig.~\ref{fig:summary}.

\section{Feedforward Control Design and Asymptotic Reachability of Stochastic Oscillators}\label{sec:reachability}
In this section, we perform the \revv{asymptotic} reachability analysis of the stationary distribution $ \bar{\rho}_{\rm st} $. 
\revv{Here, we say that the target density $ \rho_{f,0} $ is asymptotically reachable under averaging if there exists a periodic input $ u_\FF $ such that $ \bar{\rho}_{\rm st} = \rho_{f,0} $.}
Then, we propose an efficient design method of $ u_\FF $ such that $ \bar{\rho}_{\rm st} $ is closest to the target $ \rho_f $.
\subsection{Exact \revv{Asymptotic} Reachability under Averaging}\label{subsec:exact_reachable}
In this subsection, we show that the reachability of the stationary distribution $ \bar{\rho}_{\rm st} $ is characterized by Fourier coefficients of the phase sensitivity function $ Z $ and the target $ \rho_{f,0} $. \revv{In what follows, we assume that the periodic input $ u_\FF $ satisfies $ u_\FF \in L^2 (0,2\pi / \omega) $. Then, it has a Fourier series expansion
	\begin{align}
		u_\FF (t) = \sum_{k=-\infty}^{\infty} v_k \Ee^{\iu k \omega t} ,  \label{eq:periodic_input_fourier}
	\end{align}
	which converges in the $L^2$ sense.}
First, we derive a more concise expression of the stationary distribution $ \bar{\rho}_{\rm st} $ in \eqref{eq:stationary_averaged}, which is of independent interest.
For $ \Gamma $ in \eqref{eq:gamma_def}, define
    \begin{equation}
        V(\theta) := - \int_0^\theta \Gamma(\varphi) \rmd\varphi , \quad  \theta \in \calS^1 . \label{eq:gamma_integral}
    \end{equation}
The proof of the following result is given in Appendix~\ref{app:proof_gibbs}.
\begin{lemma}\label{lem:gibbs}
    Let \revv{$ u_\FF \in L^2 (0,2\pi / \omega) $} be $ 2\pi / \omega $-periodic. Let $ \{v_k\} $ and $ \{z_k\} $ be the Fourier coefficients of $ u_{\FF} $ and the phase sensitivity function $ Z $, respectively, \revv{in the sense of $ L^2 $ convergence} and assume $ v_0 = 0 $.
		Then, $ \Gamma $ defined in \eqref{eq:gamma_def} satisfies
		\begin{align}
			\Gamma (\theta)	&= \sum_{k=-\infty}^\infty z_k v_{-k} \Ee^{\iu k \theta} ,  &&\forall \theta \in \calS^1, \label{eq:gamma_fourier} 
	\end{align}
	where \revv{the series converges absolutely and uniformly on $ \calS^1 $}. Moreover, $ \bar{\rho}_{\rm st} $ defined in \eqref{eq:stationary_averaged} satisfies
    \begin{align}
		\bar{\rho}_{\rm st}(\theta) &= \frac{1}{C} \exp\left( - \frac{V(\theta)}{B^2} \right) , &&\forall \theta \in \calS^1 , \label{eq:stationary_gibbs} 
	\end{align}
    where $ C := \int_{\calS^1} \exp( - V(\theta)/B^2 ) \rmd \theta $ is the normalizing constant.
    \hfill $ \diamondsuit $
\end{lemma}

\hi{This lemma indicates that under} a periodic input $ u_\FF $ with $ v_0 = 0 $, the stationary distribution $ \bar{\rho}_{\rm st} $ takes the form of a Gibbs distribution, whose energy function is $ V $.
The following \hi{theorem is our first main result. It} shows that given a target distribution $ \rho_{f,0} $, a periodic input $ u_\FF $ with properly designed Fourier coefficients achieves $ \bar{\rho}_{\rm st} = \rho_{f,0} $.
The proof based on Lemma~\ref{lem:gibbs} is deferred to Appendix~\ref{app:exact_reachability}.
\begin{theorem}\label{thm:exact_reachability}
	Assume that the \hi{target} density $ \rho_{f,0} $ satisfies $ \rho_{f,0} \in C^1 (\calS^1) $ and $ \rho_{f,0} (\theta) > 0 $ for any $ \theta \in \calS^1 $.
	\revv{Let $ \{z_k\} $ and $ \{p_k\} $ be the Fourier coefficients of $ Z $ and $ \partial_\theta \log \rho_{f,0} $, respectively, \revv{in the sense of $ L^2 $ convergence}. That is,}
	\begin{align}
		Z(\theta) = \sum_{k=-\infty}^{\infty} z_k \Ee^{\iu k\theta}  , ~~ \partial_\theta \log \rho_{f,0} (\theta) = \sum_{k=-\infty}^\infty p_k \Ee^{\iu k \theta}  \label{eq:fourier}
	\end{align}
	hold \revv{in the sense of $ L^2 $ convergence.}
	\hi{Assume further that for any $ k_0 \in \bbZ $ such that $ p_{k_0} \neq 0 $, it holds that $ z_{k_0} \neq 0 $.}
	Define 
	\begin{align}\label{eq:fourier_input}
		v_k := 
		\begin{cases}
		\frac{B^2p_{-k}}{z_{-k}} , & k \neq 0 \ {\rm and} \ z_{-k} \neq 0 ,\\
		0, & k = 0 \ {\rm or } \ z_{-k} = 0 .
		\end{cases}
	\end{align}
	If \revv{$ \{v_k \} \in \ell^2 (\bbZ) $} and $ u_\FF $ is the $ 2\pi / \omega $-periodic input with Fourier coefficients $ \{v_k\} $ as in \eqref{eq:periodic_input_fourier}, then $ \bar{\rho}_{\rm st} $ in \eqref{eq:stationary_averaged} satisfies
	\begin{align}
		\bar{\rho}_{\rm st}(\theta) = \rho_{f,0}(\theta) , ~~ \forall \theta \in \calS^1 . \label{eq:target_stationary}
	\end{align}
	\hfill $ \diamondsuit $
\end{theorem}

Theorem~\ref{thm:exact_reachability} can be seen as the stochastic counterpart of the characterization of the reachability of deterministic oscillators in terms of Fourier coefficients given in \cite{Kuritz}.

\subsection{Approximate Reachability under Averaging}
Theorem~\ref{thm:exact_reachability} says that if the nonzero Fourier coefficients of $ \partial_\theta \log \rho_{f,0} $ are fully covered by the nonzero Fourier coefficients of the phase sensitivity function $ Z $, then we can design a periodic input $ u_\FF $ such that the stationary distribution $ \bar{\rho}_{\rm st} $ coincides with the target distribution $ \rho_{f,0} $. However, in practice, the number of nonzero Fourier coefficients of $ Z $ is typically limited, and if $ z_k $ is small, large $ v_{-k} $ is required by \eqref{eq:fourier_input}, which increases the energy of the control input.
Hence, the assumption in Theorem~\ref{thm:exact_reachability} may be restrictive.

\hi{The following theorem holds under a more relaxed assumption on the Fourier coefficients and provides several bounds on the difference between the stationary and target distributions. 
The bounds are expressed in terms of the $ L^1 $ and $ L^2 $ norms, the KL divergence, and the relative Fisher information. These bounds will play different roles in our subsequent analysis. The proof of this theorem is given in Appendix~\ref{app:appro_reachability}.}

\begin{theorem}\label{thm:appro_reachability}
	Assume that the \hi{target} density $ \rho_{f,0} $ satisfies $ \rho_{f,0} \in C^1 (\calS^1) $ and $ \rho_{f,0} (\theta) > 0 $ for any $ \theta \in \calS^1 $.
	\revv{Let $ \{z_k\} $ and $ \{p_k\} $ be the Fourier coefficients of $ Z $ and $ \partial_\theta \log \rho_{f,0} $ as in \eqref{eq:fourier} \revv{in the sense of $ L^2 $ convergence}.}
	Let \revv{$ u_{\FF} \in L^2 (0,2\pi / \omega)$} be $ 2\pi / \omega $-periodic with Fourier coefficients $ \{v_k\} \in \ell^2 (\bbZ) $ and $ v_0 = 0 $.
	Define $ a_k := \frac{z_k v_{-k}}{B^2} - p_k $.
	Then, $ \{a_k\} \in \ell^2 (\bbZ) $ and
	\begin{align}
		\|\partial_\theta \log \bar{\rho}_{\rm st} - \partial_\theta \log \rho_{f,0} \|_{2} &= (2\pi)^{1/2} \left(\sum_{k=-\infty}^\infty \left| a_k \right|^2 \right)^{1/2} , \label{eq:l2_bound} \\
		\FI{\rho_{f,0}}{\bar{\rho}_{\rm st}} &\le 2\pi \|\rho_{f,0}\|_{\infty} \sum_{k=-\infty}^\infty \left| a_k \right|^2. \label{eq:fi_bound}
	\end{align}
	Moreover, if $ \{p_k \} \in \ell^1 (\bbZ) $, then $ \{a_k \} \in \ell^1 (\bbZ) $ and
	\begin{align}
		\|\log \bar{\rho}_{\rm st} - \log \rho_{f,0} \|_{1} 
		&\le (2\pi)^2  \sum_{k=-\infty}^\infty \left| a_k \right|  , \label{eq:l1_bound} \\
				\kl{\rho_{f,0}}{\bar{\rho}_{\rm st}} &\le (2\pi)^{2} \|\rho_{f,0}\|_{\infty} \sum_{k=-\infty}^\infty \left| a_k \right| . \label{eq:kl_bound} 
	\end{align}
	\hfill $ \diamondsuit $
\end{theorem}

\hi{
	Notably, \eqref{eq:l2_bound} is an equality while \eqref{eq:fi_bound}--\eqref{eq:kl_bound} provide upper bounds.
	The periodic input $ u_\FF $ whose Fourier coefficients $ \{v_k\} $ are given by \eqref{eq:fourier_input} makes the right-hand sides of \eqref{eq:l2_bound}--\eqref{eq:kl_bound} zero. Therefore, Theorem~\ref{thm:appro_reachability} can be seen as an extension of Theorem~\ref{thm:exact_reachability} to general Fourier coefficients of $ Z $ and $ \partial_\theta \log \rho_{f,0} $. 
	The $ L^1 $ and $ L^2 $ norms are characterized by the $ \ell^1 $ and $ \ell^2 $ norms of $ \{ z_k v_{-k} B^{-2} - p_k \}_k $, respectively.
	For the error analysis of averaging in Subsection~\ref{sec:error} and the convergence analysis of feedback controllers, which will be designed in Section~\ref{sec:controller}, the KL divergence and the relative Fisher information are more tractable than the $ L^1 $ and $ L^2 $ norms. These quantities between the stationary distribution $ \bar{\rho}_{\rm st} $ and the target $ \rho_{f,0} $ can also be bounded in terms of $ \{ z_k v_{-k} B^{-2} - p_k \}_k $ as in \eqref{eq:fi_bound} and \eqref{eq:kl_bound}.
}

\subsection{Design of a Periodic Input via Convex Optimization}\label{subsec:optimization}
Inspired by Theorem~\ref{thm:appro_reachability}, we determine $ \{v_k\} $ to minimize the right-hand side of \eqref{eq:fi_bound} or \eqref{eq:kl_bound}. 
Let $ \calK := \{k\in \bbZ : p_{-k} \neq 0, \ z_{-k} \neq 0 \} $, where $ \{p_k\} $ and $ \{z_k\} $ are \hi{the} Fourier coefficients of $ \partial_\theta \log \rho_{f,0} $ and $ Z $, respectively, from \eqref{eq:fourier}. 
For $ k $ with $ p_k = 0 $, choosing $ v_{-k} = 0 $ minimizes $ | \frac{z_k v_{-k}}{B^2} - p_k | $, and for $ k $ with $ z_k = 0 $, the coefficient $ v_{-k} $ does not affect $ | \frac{z_k v_{-k}}{B^2} - p_k | $. Therefore, we set $ v_k = 0 $ for $ k \not\in \calK $ and optimize $ v_k $ for $ k\in \calK $.

By Parseval's identity, the input energy can be written as
\begin{align}
	\|u_\FF (\cdot/\omega) \|_{2}^2 = 2\pi \sum_{k=-\infty}^\infty |v_k|^2 . \nonumber
\end{align}
Then, for minimizing the difference between $ \bar{\rho}_{\rm st} $ and $ \rho_{f,0} $ under an input energy constraint, we consider the following optimization problem:
\begin{align}
  \underset{\{v_k\}_{k\in \calK}}{\text{minimize}} ~~ &\sum_{k\in \calK} \left( \left| \frac{z_{-k} v_k}{B^2} - p_{-k} \right|^{\hi{r}} + \lambda R(v_k) \right) \label{prob:opt_fourier}\\
  \text{subject~to} ~~ & \sum_{k\in \calK} |v_k|^2 \le \frac{E}{2\pi}  , \label{eq:opt_fourier_constraint}
\end{align}
where $ \hi{r} \in \{1,2\} $, $ E > 0 $, $ \lambda \ge 0 $, and $ R : \bbC \rightarrow \bbR $ is a convex regularization of $ v_k $. For example, the $ \ell^1 $~regularization $ R(v_k) = |{\rm Re}(v_k) | + |{\rm Im}(v_k) | $ promotes the sparsity of the Fourier coefficients of $ u_\FF $.
Additional constraints on $ u_{\FF} $ can be introduced to Problem~\eqref{prob:opt_fourier}.
\hi{For example,} when considering an amplitude constraint $ \|u_\FF \|_\infty \le U $, it is reasonable to replace it by a convex constraint $ \sum_{k\in \calK} |v_k| \le U $ because $ \|u_\FF \|_\infty \le \sum_{k\in \calK} |v_k| $.

To reformulate the optimization \eqref{prob:opt_fourier} over complex variables as a problem over real variables, we split $ v_k $, $ z_k $, and $ p_k $ into their real and imaginary parts as $ v_k = v_{k}^R + \iu v_{k}^I $, $ z_k = z_{k}^R + \iu z_{k}^I $, and $ p_k = p_{k}^R + \iu p_{k}^I $.
Then, Problem~\eqref{prob:opt_fourier} with $ \hi{r} = 2 $ can be rewritten as
\begin{align}
  \underset{\{v_k\}_{k\in \calK}}{\text{minimize}} ~~ &\sum_{k\in \calK} \Bigl(  (z_{-k}^R v_{k}^R - z_{-k}^I v_{k}^I - p_{-k}^R B^2 )^2 \nonumber\\
  &\quad \quad  + (z_{-k}^I v_{k}^R + z_{-k}^R v_{k}^I - p_{-k}^I B^2 )^2  + \lambda R(v_k) \Bigr) \nonumber \\
  \text{subject~to} ~~ & \sum_{k\in \calK} \left[(v_{k}^R)^2 + (v_{k}^I)^2 \right] \le \frac{E}{2\pi}  . \nonumber
\end{align}
\hi{This is} a quadratic program, and thus, can be solved efficiently. Similarly, the objective \eqref{prob:opt_fourier} with $ \hi{r} = 1 $ can be written as the following convex objective function:
\begin{align}
    &\sum_{k\in \calK} \biggl( \Bigl[  (z_{-k}^R v_{k}^R - z_{-k}^I v_{k}^I - p_{-k}^R B^2 )^2 \nonumber\\
  &\quad \quad  + (z_{-k}^I v_{k}^R + z_{-k}^R v_{k}^I - p_{-k}^I B^2 )^2 \Bigr]^{1/2}  + \lambda R(v_k) \biggr) . \nonumber
\end{align}
Moreover, since the Fourier coefficients decay fast, it is reasonable to truncate the series and use only a small number of coefficients as optimization variables.

\revv{
	\begin{rmk}\label{rmk:multi_input}
Theorems~\ref{thm:exact_reachability} and \ref{thm:appro_reachability} can be extended straightforwardly to multi-input cases, where $ Z(\theta) u(t) $ in \eqref{eq:rho} is replaced by $ \sum_{\ell=1}^m Z_\ell(\theta)u_\ell(t) $.
Assume that $Z_\ell(\theta)=\sum_{k\in\mathbb Z} z_{\ell,k}\Ee^{\iu k\theta}$ and $u_{\ell,\mathrm{FF}}(t)=\sum_{k\in\mathbb Z} v_{\ell,k}\Ee^{\iu k\omega t} $ with \(v_{\ell,0}=0\). Then, the averaged drift becomes
\[
\Gamma(\theta)
=
\sum_{k=-\infty}^\infty
\left(
\sum_{\ell=1}^m z_{\ell,k}v_{\ell,-k}
\right)\Ee^{\iu k\theta}.
\]
Consequently, the term $ z_kv_{-k}B^{-2}-p_k $ in Theorem~\ref{thm:appro_reachability} is replaced by
\begin{equation}\label{eq:multi_input}
	\frac{1}{B^2}\left(\sum_{\ell=1}^m z_{\ell,k}v_{\ell,-k} \right) -p_k .	
\end{equation}
Thus, the exact reachability condition in Theorem~\ref{thm:exact_reachability} is relaxed: for every mode $ k $ with $ p_k \neq 0 $, it suffices that $ \sum_{\ell=1}^m |z_{\ell,k}|^2>0 $.
The corresponding input design problem remains convex because \eqref{eq:multi_input} is affine in the input Fourier
coefficients.
\hfill $ \diamondsuit $
\end{rmk}
}

\subsection{Approximation Error by Averaging}\label{sec:error}
In this subsection, we analyze the error between \hi{the distribution} $ \wtilde{\rho}_\modi $ \hi{under the feedforward control in \eqref{eq:rho1-rot}} and its approximation $ \bar{\rho}_\modi $ obtained via the averaging method \hi{in \eqref{eq:rho1-mean}}.
For simplicity of analysis, we assume the following condition.
\begin{assumption}\label{ass:nondegenerate}
    The phase sensitivity function for noise is given by $ Z_w \equiv 1 $.
    \hfill $ \diamondsuit $
\end{assumption}

\hi{The analysis carried out in this subsection can be extended} to the case where $ Z_w^2 (\theta) > 0 $ for any $ \theta \in \calS^1 $. This condition means the non-degeneracy of the Fokker--Planck equation \eqref{eq:rho} and is commonly used to guarantee the existence of its solution.

The following result provides an upper bound of the error induced by the averaging method.
The proof is shown in Appendix~\ref{app:bound_average_original}.

\begin{theorem}\label{thm:bound_average_original}
    Suppose that Assumption~\ref{ass:nondegenerate} holds.
		\revv{Assume that $ \rho_{\modi,0} (\theta) > 0 $ for any $ \theta \in \calS^1 $, $ u_{\FF} \in C(\bbR) $ is $ 2\pi/\omega $-periodic, and $ \wtilde{\rho}_\modi $ following \eqref{eq:rho1-rot} satisfies}
    \begin{align*}
				M := \sup_{t\ge 0,\theta\in \calS^1} \wtilde{\rho}_\modi (t,\theta) < \infty.
    \end{align*}
		\revv{Let $ m := \inf_{t\ge 0,\theta\in \calS^1} \wtilde{\rho}_\modi (t,\theta) > 0 $.}
    Then, it holds that
    \begin{align}
        &\limsup_{t\rightarrow \infty} D_{\rm KL} (\bar{\rho}_{\rm st} (\cdot) \| \wtilde{\rho}_\modi (t,\cdot)) \nonumber\\
        &\le \frac{\sup_{s \in [0,2\pi/\omega)} \int_{\calS^1} \left(\Gamma (\theta) - A(s,\theta) \right)^2 \bar{\rho}_{\rm st} (\theta) \rmd \theta}{(m/M)^2 D^2} , \label{eq:averaging_bound}
    \end{align}
    where $ A(t,\theta) := Z(\theta + \omega t) u_\FF (t) $ and $ \Gamma(\theta) := \frac{\omega}{2\pi}\int_{0}^{2\pi/\omega}A(t,\theta) \rmd t  $.
    \hfill $ \diamondsuit $
\end{theorem}

\hi{
	It is known that non-degenerate Fokker--Planck equations exhibit contraction with respect to the KL divergence, that is, the KL divergence between two solutions with different initial densities converges to zero~\cite{Risken}.
	Contraction does not hold between solutions to \eqref{eq:rho1-rot} and the averaged equation~\eqref{eq:rho1-mean} due to their different drift coefficients.
	Nevertheless, Theorem~\ref{thm:bound_average_original} shows that an {\em approximate} contraction property holds for solutions to \eqref{eq:rho1-rot} and \eqref{eq:rho1-mean}. 
	The approximation is characterized by the numerator of the right-hand side of \eqref{eq:averaging_bound}, which represents the discrepancy between the original drift coefficient $ A $ and the averaged coefficient $ \Gamma $.
	
}

\revv{
\begin{rmk}\label{rmk:finite_time_bound}

The proof of Theorem~\ref{thm:bound_average_original} also yields a finite-time bound. Define
\begin{align*}
H_{\rm FF}(t)
&:=
D_{\rm KL}(\bar\rho_{\rm FF}(t,\cdot)
\|\wtilde{\rho}_{\rm FF}(t,\cdot)), \\
\bar{E}_{\rm FF}
&:=
\sup_{t\ge 0} \int_{\calS^1}
(\Gamma(\theta)-A(t,\theta))^2
\bar\rho_{\rm FF}(t,\theta) \rmd\theta ,
\end{align*}
and $\lambda:=(m/M)^2$. Then, by \eqref{eq:dH_FF_bound} and the Gronwall lemma, for any $\delta\in(0,D)$,
\begin{align*}
H_{\rm FF}(t)
&\le
\Ee^{-2\lambda(D-\delta)t}H_{\rm FF}(0)
\\
&\quad +
\frac{\bar{E}_{\rm FF}}
{4\lambda\delta(D-\delta)}
\left(1-\Ee^{-2\lambda(D-\delta)t}\right).
\end{align*}
Setting $\delta=D/2$ gives the following finite-time bound:
\[
H_{\rm FF}(t)
\le
\Ee^{-\lambda Dt}H_{\rm FF}(0)
+
\frac{\bar{E}_{\rm FF}}{\lambda D^2}
\left(1-\Ee^{-\lambda Dt}\right).
\]
\hfill $ \diamondsuit $
\end{rmk}

\begin{rmk}\label{rmk:assumption}
	When $Z_w\equiv1$, the diffusion coefficients of \eqref{eq:rho1-rot} and \eqref{eq:rho1-mean} coincide, and the proof of Theorem~\ref{thm:bound_average_original} only has to bound the difference of the drift coefficients $\Gamma(\theta)-A(t,\theta)$. For nonconstant
\(Z_w^2>0\), the same argument can be carried out by adding a term for the mismatch of the diffusion coefficients, for example, in \eqref{eq:mismatch_drift}. 
\hfill $ \diamondsuit $
\end{rmk}
}

\revv{
\begin{rmk}\label{rmk:tradeoff}
	There is an inherent tradeoff in the choice of the feedforward input. Although a larger input can make the stationary density of the averaged equation closer to the target density, it may increase the discrepancy between the original dynamics and its averaged approximation.
The energy bound $E$ in \eqref{eq:opt_fourier_constraint} can serve as a
parameter for balancing the closeness of the averaged stationary density to the target density and the
validity of the averaging approximation. 

For each fixed $E$, let $u_E$
be a solution to \eqref{prob:opt_fourier}--\eqref{eq:opt_fourier_constraint}.
Let $\wtilde\rho_{{\rm FF},E}$ denote the resulting distribution of
oscillators satisfying \eqref{eq:rho1-rot}, and let
$\bar\rho_{{\rm st},E}$ denote the stationary density of the averaged
equation \eqref{eq:rho1-mean}. Denote by $\alpha(E)$ the upper bound
\eqref{eq:kl_bound} on the discrepancy between the target
$\rho_{f,0}$ and $\bar\rho_{{\rm st},E}$ under $u_E$. Similarly,
denote by $\beta(E)$ the upper bound \eqref{eq:averaging_bound} on
the discrepancy between $\bar\rho_{{\rm st},E}$ and
$\wtilde\rho_{{\rm FF},E}$ under $u_E$.

Then, the Csisz\'{a}r--Kullback--Pinsker inequality
\cite{Villani, Gilardoni} and the triangle inequality yield
\[
\limsup_{t\to\infty}
\|\wtilde\rho_{{\rm FF},E}(t,\cdot)-\rho_{f,0}\|_1
\leq
\sqrt{2\alpha(E)}+\sqrt{2\beta(E)}.
\]
This bound suggests a simple procedure for determining the energy
constraint $E$: solve the convex problem \eqref{prob:opt_fourier} for
several values of $E$ and select the value that minimizes the above
bound. 
In this way, $E$ is not chosen solely to make
$\bar\rho_{{\rm st},E}$ close to $\rho_{f,0}$, but rather to balance
the error against the discrepancy between the averaged and original
dynamics.
\hfill $ \diamondsuit $
\end{rmk}

}

\section{Population-Level Feedback Control Design and Convergence Analysis}\label{sec:controller}
The previous section has investigated the design of a feedforward control $ u_\FF $ under which the periodic stationary distribution of oscillators $ \rho_{\modi} $ is close to the target $ \rho_f $ and their closeness is guaranteed.
\revv{Under the periodic feedforward input $u_{\rm FF}$ alone, the
convergence of $\rho$ to the periodic stationary distribution
$\rho_{\rm FF}$ is guaranteed to be exponential, as will be shown in
Theorem~\ref{thm:proposed}. Nevertheless, depending on the applications, faster convergence may be required. In this section, we propose a feedback controller as an additional mechanism to improve the transient response.}

\hi{In the control design}, the key idea is to compute the periodic stationary distribution offline in advance and use it as a surrogate target in designing the feedback controller instead of $ \rho_f $; see~Fig.~\ref{fig:converge}.
Here, we emphasize that $ \rho_\modi $ no longer represents the actual distribution of oscillators.
\hi{In our approach, we construct a control law based on the KL divergence as a Lyapunov functional
and then demonstrate its convergence properties.
The KL divergence is known to be useful for the convergence analysis of Fokker--Planck equations to their stationary distributions~\cite{Risken}.
It is noted that in \cite{Kuritz,Monga2018}, a control law is designed by using the $ L^2 $~distance instead.}

\begin{figure}[tb]
  \centering
  \includegraphics[scale=0.37]{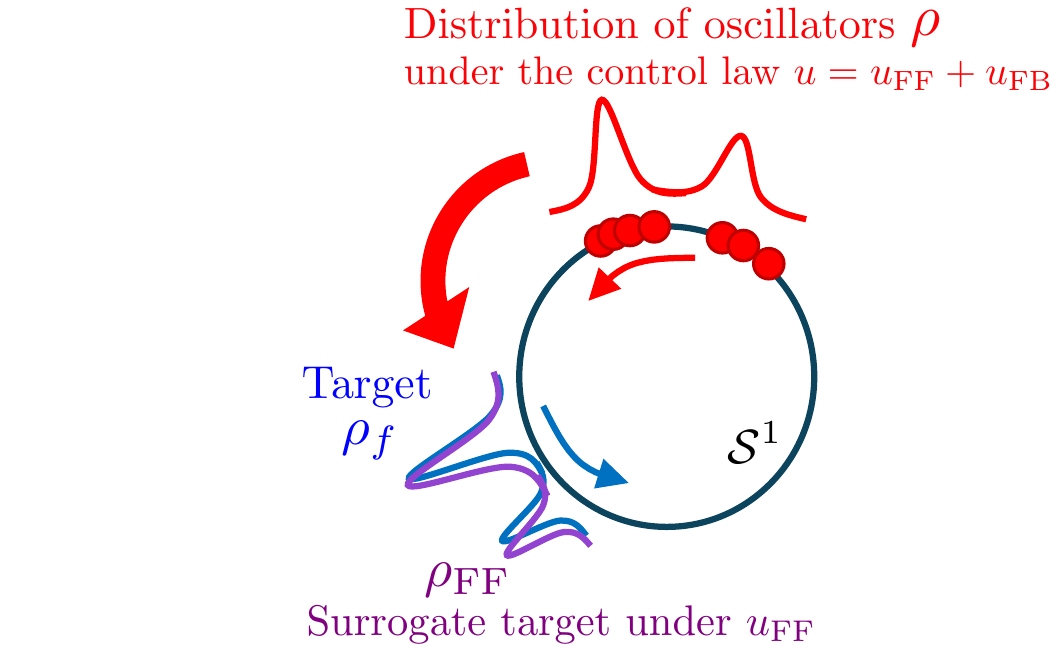}
	   \put(-210,70){\footnotesize \colorbox{mistygray}{\shortstack{Convergence of \textcolor{red}{$ \rho $} to \textcolor{plum}{$ \rho_\modi $} \\ (\hi{Theorem}~\ref{thm:proposed} and \\ \hi{Propositions}~\ref{prop:rfisher} and \ref{prop:measurement_err})}}}
  \caption{Under the proposed controller, the distribution of oscillators $ \rho $ converges to the surrogate target $ \rho_\modi $ (\hi{Theorem}~\ref{thm:proposed} and \hi{Propositions}~\ref{prop:rfisher} and \ref{prop:measurement_err}).}
  \label{fig:converge}
\end{figure}

\subsection{Design of a Feedback Controller}\label{subsec:feedback}
\revv{In Subsections~\ref{subsec:feedback} and \ref{subsec:convergence_fb}, we consider an ideal population-level
feedback setting in which the density $\rho(t,\cdot)$ of an oscillator population is available exactly to the controller.}

\hi{Our approach is to design a controller which decreases the KL divergence between the distribution of oscillators $ \rho $ and $ \rho_{\modi} $. Recall that the KL divergence is nonnegative and it holds $ \kl{\rho_1}{\rho_2} = 0 $ if and only if $ \rho_1 = \rho_2 $.}
In what follows, we assume that $ \rho_{\modi} (t,\cdot) $ is strictly positive for any \hi{$ t \ge 0 $} to guarantee the finiteness of $ H(t) := \kl{\rho(t,\cdot)}{\rho_{\modi}(t,\cdot)} $.
Similarly to \eqref{eq:derivative_H_1} in Appendix~\ref{app:bound_average_original}, the derivative of $ H(t) $ along the trajectories of \eqref{eq:rho} and \eqref{eq:rho1} is given by
\begin{align}
	\frac{\rmd H(t)}{\rmd t} = \int \rho \calL_{u}^\dagger [\log R] \rmd \theta - \int R \calL_{u_\FF}\rho_{\modi} \rmd \theta,
\end{align}
where the time $ t $ is omitted for notational simplicity, \hi{$ R(t,\theta) := \rho(t,\theta)/\rho_{\modi}(t,\theta) $,} and $ \calL_{u}^\dagger := (\omega + Z(\theta)u(t))\partial_\theta + D Z_w^2 (\theta)\partial_\theta^2 $ is the adjoint operator of $ \calL_{u} $; \hi{note that it} satisfies $ \int (\calL_u f) g \rmd \theta = \int f \calL_u^\dagger g \rmd \theta $ for any $ C^2 $~functions $ f, g $ on $ \calS^1 $~\cite{Risken}.
\hi{We can further obtain}
\begin{align}
	\frac{\rmd H(t)}{\rmd t} 
	&= \int \rho_{\modi} \calL_{u}^\dagger R\rmd \theta - \int R \calL_{u_{\FF}}\rho_{\modi}\rmd \theta \nonumber\\
    &\quad - D \int \rho Z_w^2 \left(\frac{\partial_\theta R}{R}\right)^2 \rmd \theta \nonumber\\
	&= \int R (\calL_{u}- \calL_{u_\FF})\rho_{\modi}\rmd \theta - D_{z,t} \FI{\rho_z}{\rho_{\modi,z}} \nonumber\\
	&= - \left(\int R\partial_\theta[Z\rho_{\modi}] \rmd \theta \right) (u-u_{\FF} ) \nonumber\\
	&\quad - D_{z,t}\FI{\rho_z}{\rho_{\modi,z}}, \label{eq:dHdt}
\end{align}
where
\begin{align}
    \rho_z (t,\theta) &:= \frac{\rho(t,\theta) Z_w^2(\theta)}{\int \rho Z_w^2 \rmd \theta}, ~~ \rho_{\modi,z} (t,\theta) := \frac{\rho_{\modi}(t,\theta) Z_w^2(\theta)}{\int \rho_{\modi} Z_w^2 \rmd \theta}, \nonumber\\
    D_{z,t} &:= D\left( \int \rho(t,\theta) Z_w^2(\theta) \rmd \theta  \right), \nonumber
\end{align}
and $ \int \rho Z_w^2 \rmd \theta $ and $ \int \rho_{\modi} Z_w^2 \rmd \theta $ are assumed to be nonzero. This is satisfied, for example, when $ Z_w \neq 0 $ almost everywhere.
Based on \eqref{eq:dHdt}, we obtain the following result.
\begin{proposition} \label{prop:rfisher}
	Assume that $ \rho_{\modi}(t,\cdot) $ following \eqref{eq:rho1} is strictly positive for any $ t \ge 0 $, and $ u_\FB : [0,\infty) \rightarrow \bbR $ satisfies
	\begin{equation}\label{eq:fb_cond}
		\left(\int_{\calS^1} \frac{\rho(t,\theta)}{\rho_{\modi}(t,\theta)} \partial_\theta [Z\rho_{\modi}(t,\theta)] \rmd \theta \right) u_\FB (t) \ge 0, ~~ \forall t \ge 0 .
	\end{equation}
		Also, assume that $ D_{z,t} $ is uniformly positive and $ \rmd H / \rmd t $ is uniformly continuous.
	Then, under $ u = u_\FF + u_\FB $, the relative Fisher information $ \FI{\rho_z (t,\cdot)}{\rho_{\modi,z}(t,\cdot)} $ converges to $ 0 $ as $ t \rightarrow \infty $ for any initial densities $ \rho_0 $, $ \rho_{\modi,0} $ and any $ u_{\FF} $.
	\hfill $ \diamondsuit $
\end{proposition}
\begin{proof}
	Under the condition~\eqref{eq:fb_cond}, the first term of \eqref{eq:dHdt} is nonpositive.
	Therefore, it holds that
	\begin{align}\label{eq:dissipation}
	\revv{\frac{\rmd H(t)}{\rmd t} \le - D_{z,t}\FI{\rho_z}{\rho_{\modi,z}} \le 0 .}
\end{align}
	Hence, $ H(t) $ is monotonically decreasing and converges to a nonnegative constant. Thus, it follows from Barbalat's lemma that $ \rmd H(t) / \rmd t $ converges to $ 0 $ as $ t\rightarrow \infty $ due to its uniform continuity.
	The desired result follows from $ \FI{\rho_z}{\rho_{\modi,z}} \ge 0 $, the uniform positivity of $ D_{z,t} $, and \eqref{eq:dissipation}.
\end{proof}

By the above result, under a control input $ u = u_\FF + u_\FB $ with $ u_\FB $ satisfying \eqref{eq:fb_cond}, the distribution $ \rho $ converges to the surrogate target $ \rho_{\modi} $ in the sense that $ \FI{\rho_z}{\rho_{\modi,z}} \rightarrow 0 $.
Since the initial distribution $ \rho_{\modi,0} $ of $ \rho_\modi $ is allowed to be different from that of oscillators $ \rho_0 $, the precomputed periodic stationary distribution of $ \rho_\modi $ can be used as $ \rho_{\modi,0} $.

The condition \eqref{eq:fb_cond} is satisfied, for example, by the population-level feedback control
\begin{align}
  u_\FB (t) &= u_{\rho_{\modi}} [\rho(t,\cdot)] \nonumber\\
	&:= k\int_{\calS^1} \frac{\rho(t,\theta)}{\rho_{\modi}(t,\theta)} \partial_\theta[Z\rho_{\modi}(t,\theta)] \rmd \theta, \quad k > 0. \label{eq:sqinput}
\end{align}
Moreover, when considering a constraint on the amplitude of the input $ \lw(t) \le u(t) \le \up(t) $, the following control law also decreases $ \FI{\rho_z}{\rho_{\modi,z}} $ to zero because \eqref{eq:fb_cond} is a condition on the sign of $ u_\FB $\hi{:}
\begin{align}
	u(t) &= {\rm sat}_{\lw (t)}^{\up (t)} \left(u_\FF(t) + u_{\rho_{\modi}} [\rho(t,\cdot)]\right), \label{eq:proposed}\\
	{\rm sat}_{a}^{b} (u) &:=
	\begin{cases}
		a, & u < a, \\
		u, & a \le u \le b, \\
		b, & u > b ,
	\end{cases}
	\label{eq:saturation} \\
	\lw(t) &\le u_\FF (t) \le \up(t), \quad \forall t > 0 . \nonumber
\end{align}
In summary, we propose to use \eqref{eq:proposed} combining the feedforward input $ u_\FF $ and the feedback control $ u_{\FB} = u_{\rho_{\modi}} $.

\revv{Compared with the case without feedback control $u_{\rm FB}\equiv0$, the proposed feedback input adds an extra nonpositive term, which is the first term of \eqref{eq:dHdt}, to the instantaneous decrease rate of the KL divergence. The acceleration effect of this additional dissipation term is examined numerically in Section~\ref{sec:example}.}

\revv{
\begin{rmk}\label{rmk:sensitivity_rhoFF}
	Since the feedback input \(u_{\mathrm{FB}}\) in \eqref{eq:sqinput} contains \(\rho_{\mathrm{FF}}\) in the denominator, numerical errors in computing the surrogate target \(\rho_{\mathrm{FF}}\) may be amplified when \(\rho_{\mathrm{FF}}\) is small in a region where the oscillator density \(\rho\) is relatively large. The sensitivity of \(u_{\mathrm{FB}}\) with respect to \(\rho_{\mathrm{FF}}\) can be characterized as follows. By rewriting
\[
\frac{\partial_\theta[Z(\theta)\rho_{\mathrm{FF}}(t,\theta)]}{\rho_{\mathrm{FF}}(t,\theta)}
=
\partial_\theta Z(\theta)
+
Z(\theta)\partial_\theta\log\rho_{\mathrm{FF}}(t,\theta),
\]
we see that the relevant quantity is the logarithmic derivative \(\partial_\theta\log\rho_{\mathrm{FF}}\). To quantify the sensitivity, let $ \rho_\modi^{(1)} $ and $ \rho_\modi^{(2)} $ be two reference densities, and denote the corresponding feedback inputs by $ u_\FB^{(1)} $ and $ u_\FB^{(2)} $, respectively.
Then, we have
\begin{align*}
&\left| u_\FB^{(1)} (t) - u_\FB^{(2)} (t)      \right| \\
&= k \left| \int_{\calS^1} Z(\theta) \! \left(\partial_\theta \log \rho_\modi^{(1)} (t,\theta)- \partial_\theta \log \rho_\modi^{(2)} (t,\theta) \right) \rho(t,\theta) \rmd \theta \right| \\
&\le k \left\| Z \! \left(\partial_\theta \log \rho_\FF^{(1)} (t,\cdot) - \partial_\theta \log \rho_\FF^{(2)} (t,\cdot)\right) \right\|_\infty .
\end{align*}
This shows that the sensitivity of the feedback control to perturbations of $ \rho_{\FF} $ is governed by the variation of $ \partial_\theta \log \rho_\FF $ weighted by $ Z $.
\hfill $ \diamondsuit $
\end{rmk}
}

\subsection{Further Convergence Analysis}\label{subsec:convergence_fb}
\hi{As an extension of Proposition~\ref{prop:rfisher}, we next} present the convergence result for the proposed method in the sense of the KL divergence rather than the relative Fisher information and derive its consequences. 
\rev{Let $W_2(\rho_1,\rho_2)$ be the 2-Wasserstein distance between densities $ \rho_1, \rho_2 $ on $ \calS^1 $ endowed with a distance~\cite{Villani}.}
\revv{By assuming the non-degeneracy condition $ Z_w^2 > 0 $, we can relate the upper bound of \eqref{eq:dissipation} to the KL divergence $ H(t) $, and consequently, the convergence in the KL divergence is guaranteed as follows.}

\begin{theorem}\label{thm:proposed}
	Assume that $ Z_w^2 (\theta)> 0 $ and \revv{$ \rho_{\modi,0} (\theta) > 0 $} for any $ \theta \in \calS^1 $. 
	Let $ u_{\FF} \in C(\bbR) $ be $ 2\pi/\omega $-periodic.
  Also, assume that $\rho_{\modi}(t,\cdot)$ following \eqref{eq:rho1} is uniformly bounded in $t$.
	Then, under a control input $ u = u_\FF + u_\FB $ whose $ u_\FB $ satisfies \eqref{eq:fb_cond}, the following properties hold for any initial density $ \rho_0 $:
	\begin{enumerate}
		\item $\kl{\rho(t,\cdot)}{\rho_{\modi}(t,\cdot)}$ converges to $ 0 $ as $ t\rightarrow\infty $.
		\item $ \|\rho(t,\cdot) - \rho_\modi (t,\cdot) \|_1 $ converges to $ 0 $ as $ t\rightarrow\infty $.
		\item $ W_2(\rho(t,\cdot), \rho_{\modi}(t,\cdot)) $ converges to $ 0 $ as $ t\rightarrow\infty $.
	\end{enumerate}
		In addition, the above convergence rates are exponential.
    \hfill $ \diamondsuit $
\end{theorem}

\revv{
To clarify the role of the assumptions, we provide a proof sketch below; see~Appendix~\ref{app:proposed} for the complete proof.
First, we recall the functional inequality used below.
\begin{definition}\label{def:sobolev}
	A density function $ \rho_2 \in C^1(\calS^1) $ is said to satisfy the logarithmic Sobolev inequality ($ {\rm LSI}(\lambda) $) with $ \lambda > 0 $ if for any density $ \rho_1 \in C^1(\calS^1) $, it holds that
	\begin{equation}\label{eq:log_Sobolev}
		\FI{\rho_1}{\rho_2} \ge 2\lambda \kl{\rho_1}{\rho_2}  .
	\end{equation}
	\hfill $ \diamondsuit $
\end{definition}

If a density $ \rho_2 $ satisfies
\[
0<m\le \rho_2(\theta)\le M<\infty,
\quad \theta\in \calS^1,
\]
for some constants $ m $ and $ M $, then $ \rho_2 $ satisfies $ {\rm LSI} ((m/M)^2) $; see Lemma~\ref{lem:lsi_bounded} for details.

}

\revv{
	\emph{Proof sketch of Theorem~\ref{thm:proposed}.}
By the positivity of $ Z_w^2 $, there exists a constant $ c_1 > 0 $ such that $\FI{\rho_z}{\rho_{\modi,z}} \ge c_1 \FI{\rho}{\rho_\modi} $.
Then, it follows from \eqref{eq:dissipation} that
    \begin{align}\label{eq:H_F}
        \frac{\rmd H(t)}{\rmd t} \le - c_2  \FI{\rho}{\rho_{\modi}} 
    \end{align}
		for some $ c_2 > 0 $.
Moreover, the uniform positivity $ \rho_{\FF} $ is guaranteed under the initial positivity $ \rho_{\modi,0} > 0 $, the periodicity and continuity of $ u_{\FF} $, and the non-degeneracy of the diffusion coefficient~\cite[Corollary~3.1]{Bogachev2009positive}.
Combining this with the uniform boundedness of $ \rho_{\FF} $ implies that $ \rho_\modi $ satisfies the logarithmic Sobolev inequality~\eqref{eq:log_Sobolev}. Thus, it follows from \eqref{eq:H_F} that
\begin{equation}
\frac{\rmd H(t)}{\rmd t} \leq - c_3 H(t) 
\end{equation}
for some $ c_3 > 0 $. This means the exponentially fast convergence $ H(t) \le \Ee^{- c_3 t} H(0) $.

The exponential convergence in the $ L^1 $ and $ W_2 $ senses follows from this bound through the Csisz\'{a}r--Kullback--Pinsker inequality and Talagrand's inequality~\cite{Otto}:
\begin{align}
	\|\rho-\rho_{\modi}\|_1^2 &\leq 2H(t), \nonumber\\
	W_2^2(\rho(t,\cdot), \rho_{\modi}(t,\cdot)) &\leq c_4 H(t) ,  \ \ c_4 > 0.
\end{align}
\hfill $\blacksquare$
}

The exponential convergence rates of $ \rho $ depend on the noise intensity $ D $ and $ \lambda = (m/M)^2 $, where $ m $ and $ M $ are defined in Theorem~\ref{thm:bound_average_original}. For larger $ \lambda $, the convergence becomes faster. The maximum value of $ \lambda $ is attained when $ m = M $, that is, $ \rho_{\modi} $ is uniform. Therefore, $ \rho_{\modi} $ with smaller up-down swings results in the faster convergence of $ \rho $.

The assumption $ Z_w^2 > 0 $ in \hi{Theorem}~\ref{thm:proposed} can be removed by assuming instead that there exists $ \lambda > 0 $ such that $ \rho_{\modi,z} $ satisfies $ {\rm LSI}(\lambda) $ for any $ t > 0 $ and replacing $ \rho $ and $ \rho_\modi $ in the statements 1)--3) by $ \rho_z $ and $ \rho_{\modi,z} $, respectively.

\hi{Note that the condition \eqref{eq:fb_cond} for the feedback control $ u_{\FB} $ in \hi{Proposition}~\ref{prop:rfisher} and \hi{Theorem}~\ref{thm:proposed} is satisfied even for $ u_{\FB} \equiv 0 $, which corresponds to the case where only a periodic feedforward input $ u_\FF $ is present. Hence, $ \rho $ converges to $ \rho_{\modi} $ irrespective of the initial distributions $ \rho_{0} $ and $ \rho_{\modi,0} $ by using only the periodic feedforward input $ u = u_\FF $. This periodic control input can also be based on the previous work~\cite{Kato}, but convergence to the desired distribution has not been established there.
In addition to the convergence, the proposed controller \eqref{eq:proposed} with $ k > 0 $ improves the transient response of $ \rho $ compared to \cite{Kato}, as we will see in Section~\ref{sec:example}.}

\revv{
	\begin{rmk}\label{rmk:uncontrollable}
We emphasize the distinction between convergence to the
surrogate target $\rho_{\rm FF}$ and direct stabilization to the
original target $\rho_f$. For the surrogate target, if $ \int_{\calS^1} \frac{\rho(t,\theta)}{\rho_{\FF}(t,\theta)} \partial_\theta [Z(\theta)\rho_\FF(t,\theta)] \rmd \theta = 0 $, then by \eqref{eq:dHdt}, the feedback input cannot add extra
instantaneous dissipation to the KL Lyapunov functional at that time.
Nevertheless, the convergence of $ \rho $ to $ \rho_\FF $ still holds due to the nonpositive relative Fisher information, as established in Proposition~\ref{prop:rfisher} and
Theorem~\ref{thm:proposed}.

On the other hand, direct stabilization to the original target
$\rho_f$ is different. The target $\rho_f$ is not a solution to the
same Fokker--Planck equation that $ \rho $ satisfies, and a direct calculation yields
\begin{align}
    &\frac{\rmd }{\rmd t} \kl{\rho}{\rho_f} \nonumber\\
    &\le - \left( \int_{\calS^1} \frac{\rho}{\rho_f} \partial_\theta [Z \rho_f] \rmd \theta \right) u \nonumber\\
    &\quad- D_{z,t} \FI{\rho_z}{\rho_{f,z}}^{\frac{1}{2}} \nonumber\\
    &\quad \times \left[ \FI{\rho_z}{\rho_{f,z}}^{\frac{1}{2}} - \sqrt{2\pi} \|\rho_z\|_\infty \FI{\frac{1}{2\pi}}{\rho_{f,z}}^{\frac{1}{2}}  \right] , \label{eq:rho_rhof}
\end{align}
where $ \rho_{f,z} (t,\theta) := \frac{\rho_f(t,\theta) Z_w^2(\theta)}{\int \rho_f Z_w^2 \rmd \theta} $.
The detailed derivation is given in Appendix~\ref{app:rho_rhof}. 
Unlike the case of $\rho_{\rm FF}$, the second term is not
sign-definite in general. Therefore, if $\int_{\calS^1} \frac{\rho}{\rho_f} \partial_\theta [Z \rho_f] \rmd \theta=0$, the control input cannot contribute to the instantaneous descent of the KL divergence towards
$\rho_f$, while the second term~\eqref{eq:rho_rhof} does not necessarily decrease the
KL divergence. This observation is one reason for introducing the
surrogate target $\rho_{\rm FF}$, for which convergence can be
guaranteed by the contraction property of the Fokker--Planck equation.
\hfill $ \diamondsuit $
	\end{rmk}
}

\subsection{Feedback Control under Measurement Errors}\label{subsec:measurement}
We have seen that the proposed feedback controller can steer the distribution of oscillators to the surrogate target $ \rho_\modi $ with the assumption that the exact measurement of the distribution $ \rho $ is possible.
However, in practice, \hi{measurement} errors are inevitable and neglecting them may prevent the convergence of $ \rho $ to $ \rho_\modi $.
In this subsection, we modify the feedback controller \eqref{eq:sqinput} so that the convergence is guaranteed even under measurement errors.

\hi{Specifically}, we consider the situation where, in the presence of measurement errors, an estimate $ \what{\rho} $ is used instead of the true distribution $ \rho $.
By replacing $ \rho $ by $ \what{\rho} $ in \eqref{eq:sqinput}, we modify the controller as follows:
\begin{align}
	\what{u}_\FB (t) := k \int_{\calS^1} \frac{\what{\rho} (t,\theta)}{\rho_{\modi} (t,\theta)} \partial_\theta [Z(\theta)\rho_{\modi} (t,\theta)] \rmd \theta . \label{eq:feedback_estimate}
\end{align}
By \eqref{eq:dHdt}, under $ u(t) = u_\FF  (t) + \what{u}_\FB (t) $, it holds that
\begin{align}
	\frac{\rmd \kl{\rho}{\rho_{\modi}}}{\rmd t} &= - k \left( \int \frac{\rho}{\rho_{\modi}} \partial_\theta [Z\rho_{\modi}] \rmd \theta \right) \nonumber\\
	&\quad \times \left(\int  \frac{\what{\rho}}{\rho_{\modi}}\partial_\theta[Z\rho_{\modi}] \rmd \theta \right) \nonumber\\
    &\quad - D_{z,t} \FI{\rho_z}{\rho_{\modi,z}} . \label{eq:KL_FF2}
\end{align}
When $ \what{\rho} \neq \rho $, the first term induced by control can be positive, unlike the case $ \rho = \what{\rho} $.
In order to guarantee that the feedback control~\eqref{eq:feedback_estimate} always speeds up the convergence of $ \rho $, when the first term of \eqref{eq:KL_FF2} is ensured to be nonpositive, we inject $ u(t) = u_\FF (t) + \what{u}_\FB (t) $ into the system~\eqref{eq:rho}, and otherwise, we only apply the feedforward input $ u (t) = u_\FF (t) $.

\hi{In the following proposition, we show that under} the assumption that the error between $ \rho $ and $ \what{\rho} $ is bounded as $  \| \rho - \what{\rho} \|_2 \le e $, a sufficient condition for the first term of \eqref{eq:KL_FF2} to be nonpositive is given by
\begin{align}
    \frac{|\what{u}_\FB(t)|}{k} &= \left|\int \frac{\what{\rho}(t,\theta)}{\rho_{\modi} (t,\theta)} \partial_\theta [Z(\theta)\rho_{\modi} (t,\theta)] \rmd \theta \right| \nonumber\\
    &\ge e \Biggl[ 
        \int_{\calS^1} \left( 
        \frac{\partial_\theta [Z(\theta)\rho_{\modi}(t,\theta)]}{\rho_{\modi}(t,\theta)}  
        \right)^2 \rmd \theta  
        \Biggr]^{1/2} . \label{eq:decrease_cond}
\end{align}
The proof is deferred to Appendix~\ref{app:measurement_err}.
\begin{proposition}\label{prop:measurement_err}
    Assume that $ \| \rho (t,\cdot) - \what{\rho} (t,\cdot) \|_2 \le e  $ for any $ t \ge 0 $.
    Let a control input be given by $ u = u_\FF + u_\FB $ whose $ u_{\FB} $ satisfies for any $ t \ge 0 $,
    \begin{align}
        u_\FB(t) =
        \begin{cases}
    \what{u}_\FB (t), & \text{if \eqref{eq:decrease_cond} holds} , \\
    0, & \text{otherwise}.
\end{cases} \label{eq:proposed_err}
    \end{align} 
    Then, under the assumptions of Proposition~\ref{prop:rfisher}, $ \FI{\rho_z(t,\cdot)}{\rho_{\modi,z}(t,\cdot)} $ converges to $ 0 $ as $ t \rightarrow \infty $ for any initial densities $ \rho_0 $, $ \rho_{\modi,0} $ and any $ u_{\FF} $. On the other hand, under the assumptions in \hi{Theorem}~\ref{thm:proposed}, \hi{the statements} 1)--3) of \hi{Theorem}~\ref{thm:proposed} hold.
    \hfill $ \diamondsuit $
\end{proposition}

To give the intuition behind the condition \eqref{eq:decrease_cond}, we bound $ \what{u}_\FB $ as follows:
\begin{align}
	\left|\what{u}_\FB (t) \right| &= k \left| \int \frac{\what{\rho} - \rho_{\modi}}{\rho_{\modi}} \partial_\theta(Z\rho_{\modi}) \rmd \theta \right| \nonumber\\
	&\le k \left[ \int \left( \frac{\partial_\theta (Z\rho_{\modi})}{\rho_{\modi}}  \right)^2 \rmd \theta  \right]^{1/2} \| \what{\rho} - \rho_{\modi} \|_{2} . \nonumber
\end{align}
Hence, if \eqref{eq:decrease_cond} holds, we have $ \| \what{\rho} - \rho_{\modi} \|_2 \ge e $.
This implies that the closer the estimate $ \what{\rho} $ is to $ \rho_{\modi} $, the smaller the $ L^2 $~distance between the true distribution $ \rho $ and its estimate $ \what{\rho} $ must be for \eqref{eq:decrease_cond} to be satisfied.
This can be understood as indicating that more accurate measurements are required as $ \rho $ gets closer to the target distribution so that the feedback controller contributes to speeding up the convergence of $ \rho $ to $ \rho_{\modi} $.

\section{Numerical Example}\label{sec:example}
In this section, we illustrate the effectiveness of the proposed method via a numerical example.
We consider the \rev{FitzHugh--Nagumo} model used to describe the action potential of a neuron~\cite{FitzHugh,Nagumo,Kato}:
\begin{align*}
	&\rmd x = (x - ax^3 - y + u)\rmd t + \sqrt{2D} \rmd W, \\
	&\rmd y = \eta(x + b) \rmd t ,
\end{align*}
where we set $ a = 1/3 $, $b = 1/4$, $ \eta = 1/4 $, and $ D = 0.007 $.
Then, by using the phase reduction method (see\hi{,} e.g., \cite[Subsection~3.4]{Nakao2016}), we obtain $ \omega = 0.4034 $ and the phase sensitivity function $ Z = Z_w $ as shown in Fig.~\ref{fig:Z_wcauchy}\subref{fig:FitzZ}. 
As a desired shape of the distribution of oscillators, we employ a wrapped Cauchy distribution on $ \calS^1 $~\cite{Mardia}:
\begin{align}
  w(\theta, \mu, \gamma) := \frac{1}{2\pi} \frac{\mathrm{sinh}(\gamma)}{\cosh(\gamma) - \cos(\theta-\mu)} ,\label{eq:wcauchy}
\end{align}
where $ \mu \in \calS^1 $ is a location parameter, and $ \gamma > 0 $ is a scale parameter.
Fig.~\ref{fig:Z_wcauchy}\subref{fig:wcauchy} illustrates \rev{two cases} $ w(3\theta,0,1) $ and $ w(\theta,\pi,0.5) $. The desired shape is set to $ \rho_{f,0}(\theta) = w(3\theta,0,1) $, which divides oscillators into three clusters, and the initial density of oscillators is set to $ \rho_{0}(\theta) = w(\theta,\pi,0.5) $.
All numerical simulations were conducted on a MacBook Pro equipped with an Apple M3 Max chip. The simulations were implemented in MATLAB R2024a, utilizing the CVX package for convex optimization~\cite{cvx,gb08}.

\begin{figure}[t]
	\begin{minipage}[b]{0.5\linewidth}
		\centering
		\includegraphics[keepaspectratio, width=\linewidth]
		{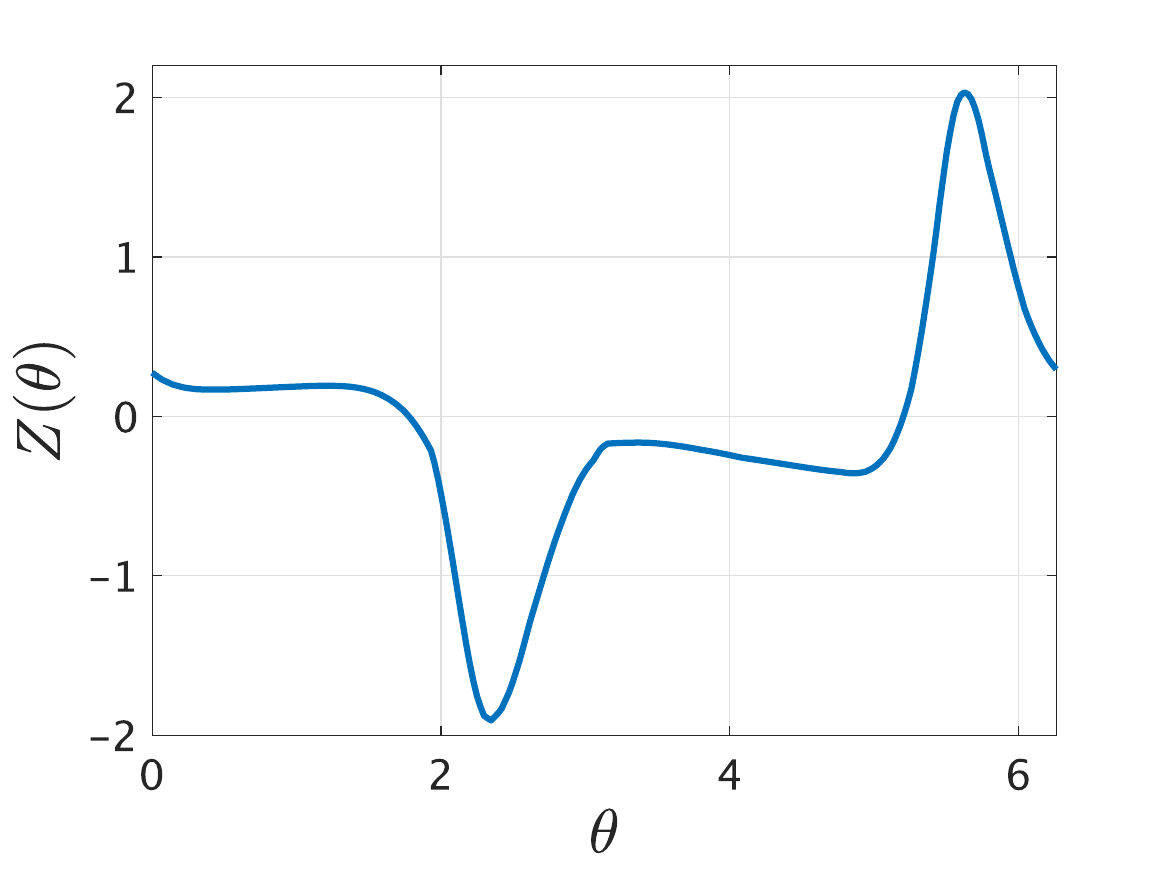}
		\subcaption{}\label{fig:FitzZ}
	\end{minipage}
	\begin{minipage}[b]{0.49\linewidth}
		\centering
		\includegraphics[keepaspectratio, width=\linewidth]
		{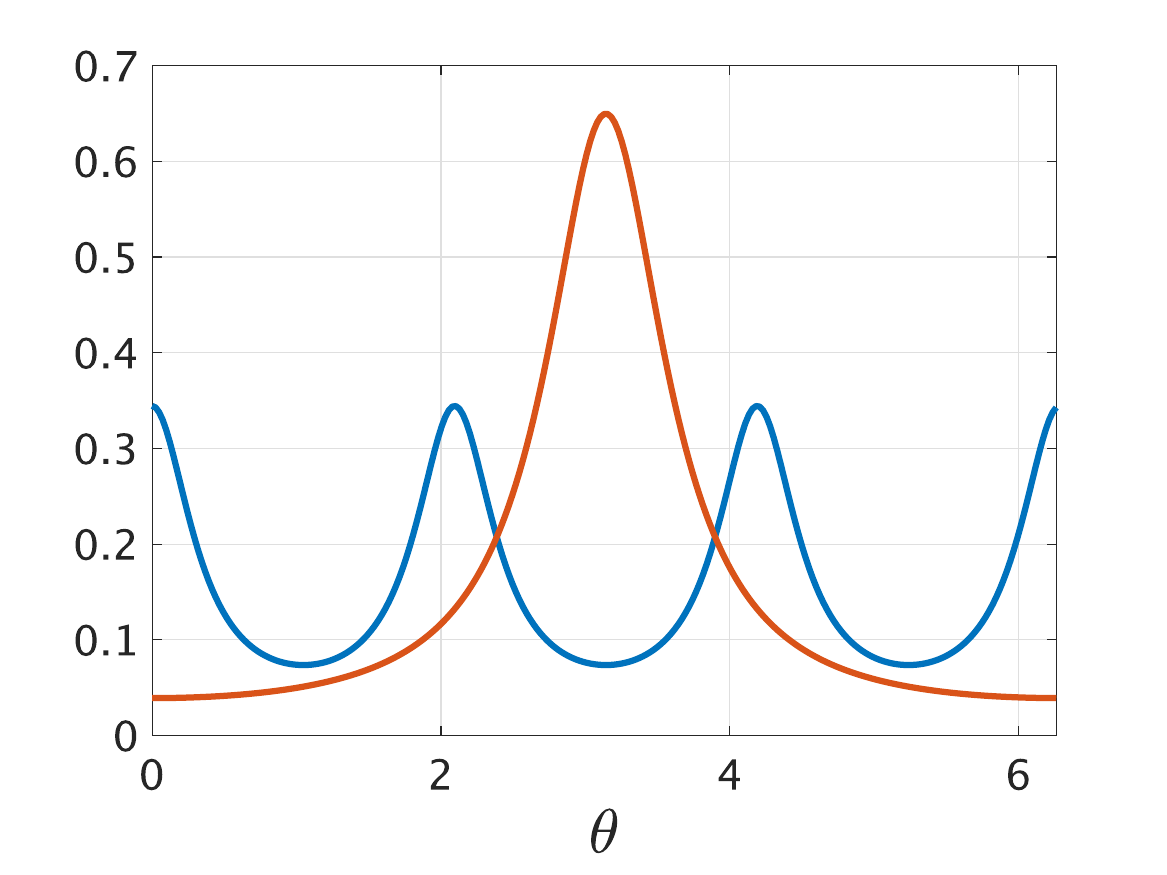}
		\subcaption{}\label{fig:wcauchy}
    \end{minipage}
	\caption{\subref{fig:FitzZ} Phase sensitivity function $Z(\theta)$ for the FitzHugh--Nagumo model. \subref{fig:wcauchy} Wrapped Cauchy distributions $ w(3\theta,0,1) $ (blue), $ w(\theta,\pi,0.5) $ (red).}\label{fig:Z_wcauchy}
\end{figure}

\subsection{\hi{Periodic Input Design Method}}
First, we design a periodic input $ u_\FF $ by solving the optimization problem~\eqref{prob:opt_fourier} with $ E = 0.02 $.
We truncate the Fourier coefficients and optimize $ \{ v_k\}_{k=-20}^{20} $.
Figs.~\ref{fig:u_ff}\subref{fig:u_ff_sub},~\subref{fig:u_fourier} illustrate the optimal inputs $ u_\FF $ and their Fourier coefficients $ \{v_k\} $ that solve \eqref{prob:opt_fourier} without regularization ($ \lambda = 0 $) and with $ \lambda = 10^{-5} $ and $ R(v) = |{\rm Re}(v)| + |{\rm Im}(v)| $, respectively.
As can be seen in Figs.~\ref{fig:u_ff}\subref{fig:u_fourier}, \subref{fig:l2_l1_compare}, the $ \ell^1 $-regularization $ R $ promotes the sparsity of the Fourier coefficients without significantly degrading the approximation accuracy of the target $ \rho_{f,0} $. 
\hi{Fig.~\ref{fig:snapshot_target} shows that the shape of the periodic stationary distribution $ \rho_{\modi} $ obtained by solving \eqref{prob:opt_fourier} with $ E = 0.02 $ and $ \lambda = 0 $ is close to the target $ \rho_{f,0} $.}

\hi{For} comparison with the existing method \cite{Kato} for determining $ u_{\FF} $, we also solve \eqref{prob:fok-opt}, which is shown in Fig.~\ref{fig:u_ff}\subref{fig:u_ff_sub}. We can observe that $ u_\FF $ obtained by \hi{this} method is similar to the one calculated by our method without regularization. 
For solving \eqref{prob:fok-opt}, we reduced it to a finite-dimensional problem by discretizing $ u_\FF $ in time and used MATLAB’s {\tt fmincon} to solve the nonconvex optimization problem~\eqref{prob:fok-opt}, which took 81.6 seconds. In contrast, the computation time for our method is only 2.1 seconds, showing its computational efficiency.

\begin{figure}[tb]
	\begin{minipage}[b]{0.5\linewidth}
		\centering
		\includegraphics[keepaspectratio, scale=0.24]
		{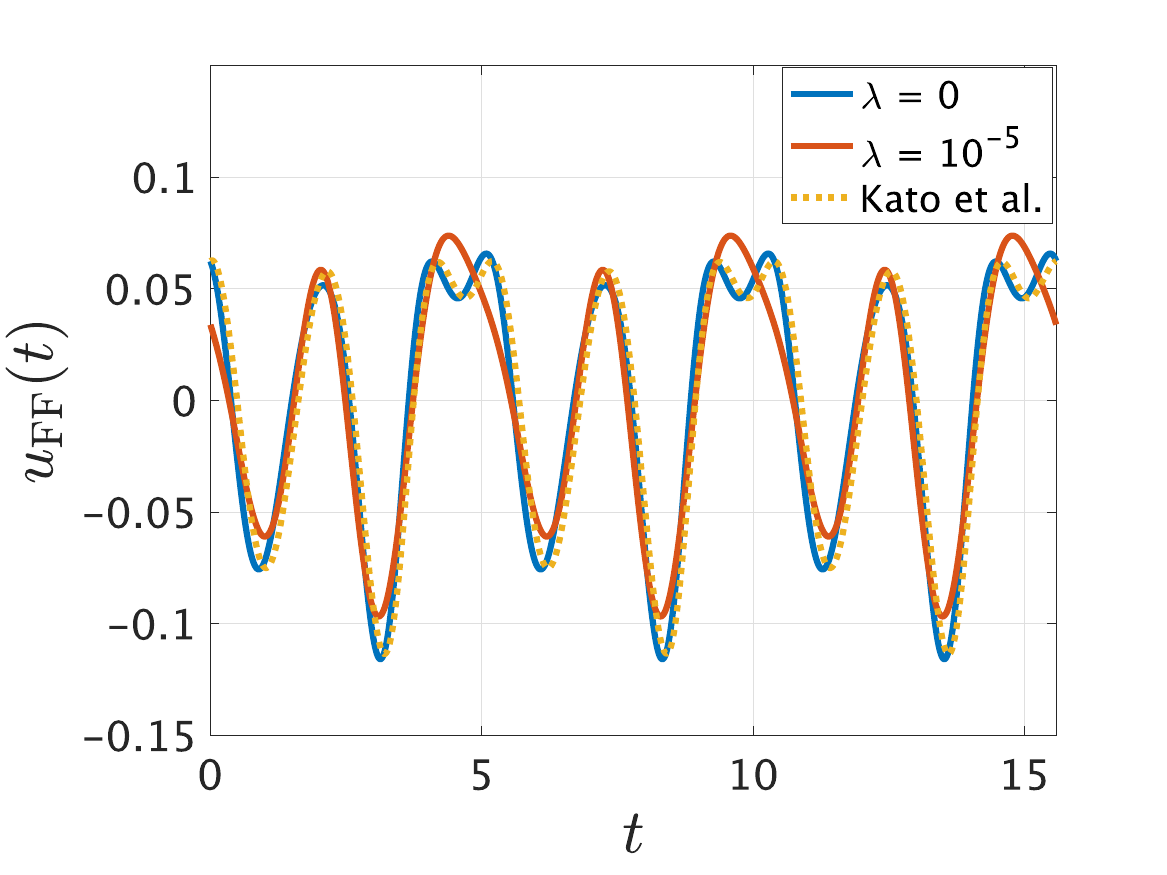}
		\subcaption{}\label{fig:u_ff_sub}
	\end{minipage}
	\begin{minipage}[b]{0.49\linewidth}
		\centering
		\includegraphics[keepaspectratio, scale=0.24]
		{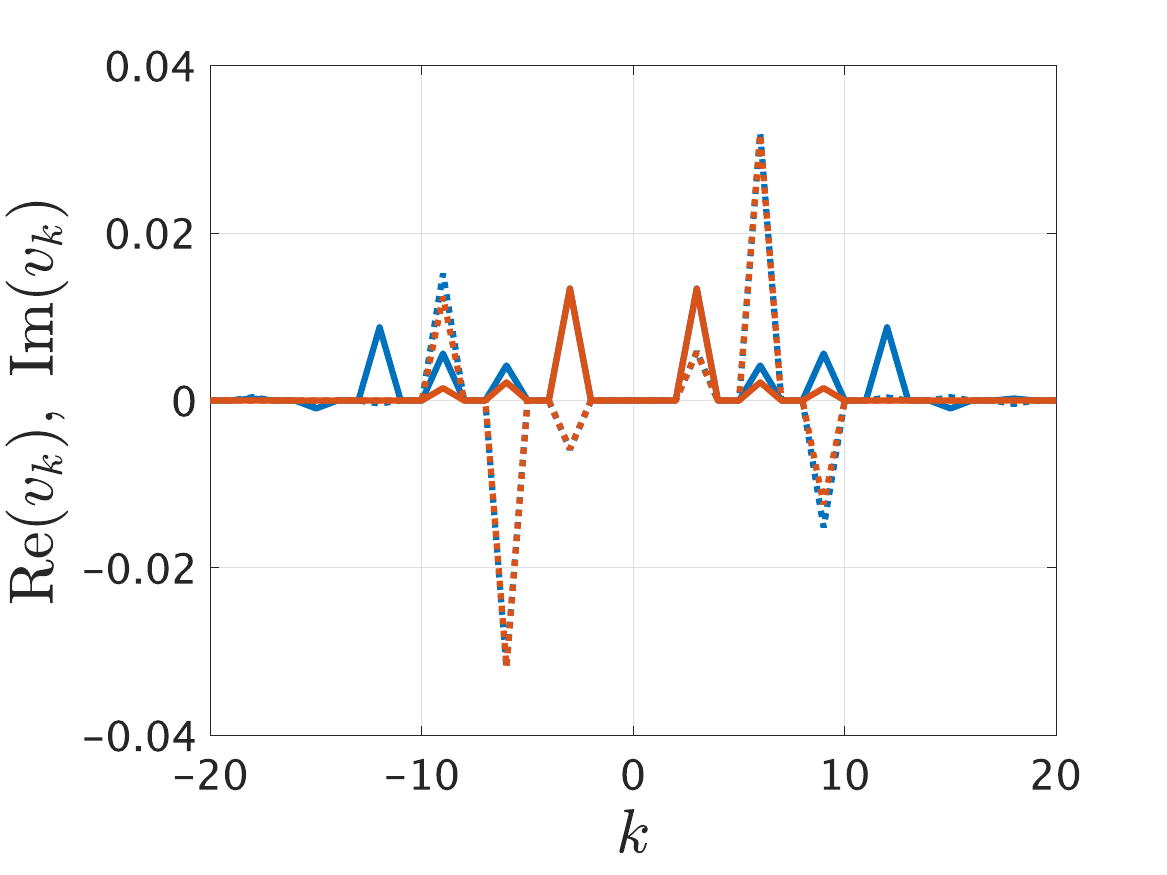}
		\subcaption{}\label{fig:u_fourier}
	\end{minipage}
    \begin{minipage}[b]{1.0\linewidth}
		\centering
		\includegraphics[keepaspectratio, scale=0.33]
		{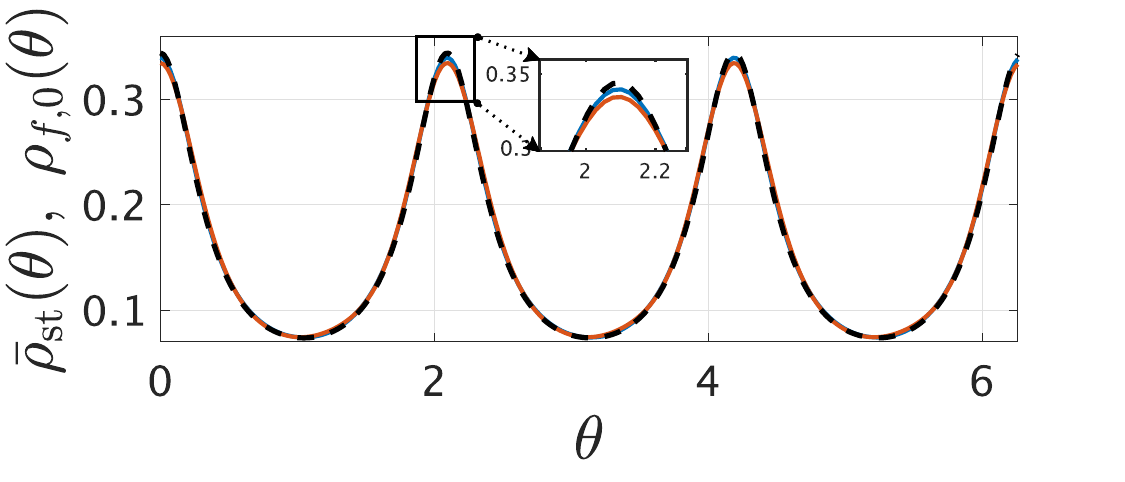}
		\subcaption{}\label{fig:l2_l1_compare}
	\end{minipage}
	\caption{\subref{fig:u_ff_sub} Periodic inputs $ u_\FF $ obtained by solving \eqref{prob:opt_fourier} with $ \hi{r} = 2 $, $ R(v) = |{\rm Re}(v)| + |{\rm Im}(v)| $, and $ \lambda = 0 $ (blue) and $ \lambda = 10^{-5} $ (red) and by solving \eqref{prob:fok-opt} (yellow, dotted). \subref{fig:u_fourier} Fourier coefficients $ v_k $ of $ u_\FF $. Blue (resp. red) solid and dotted lines show $ {\rm Re}(v_k) $ and $ {\rm Im } (v_k) $, respectively, obtained by \eqref{prob:opt_fourier} with $ \lambda = 0 $ (resp. $ \lambda = 10^{-5} $). \subref{fig:l2_l1_compare} Target density $ \rho_{f,0} $ (black, dashed) and the stationary densities $ \bar{\rho}_{\rm st} $ under averaging for $ \lambda = 0 $ (blue) and $ \lambda = 10^{-5} $ (red).}\label{fig:u_ff}
\end{figure}

\begin{figure}[tb]
  \centering
  \includegraphics[scale=0.35]{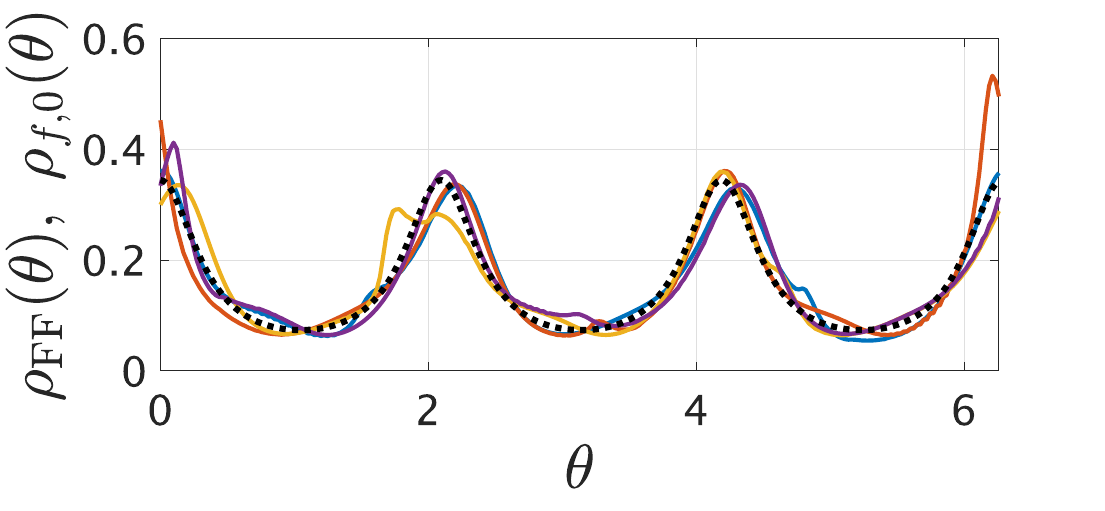}
  \caption{Target shape $ \rho_{f,0}(\theta) = w(3\theta,0,1) $ (dotted) and snapshots of the surrogate target $ \rho_{\modi}(t,\theta) $ at every $ 1/4 $ period (solid) under the coordinate transformation $\theta\mapsto\theta-\omega t $.}
  \label{fig:snapshot_target}
\end{figure}

Fig.~\ref{fig:rho_st_compare_E} compares the target distribution $ \rho_{f,0} $ with the stationary distributions $ \bar{\rho}_{\rm st} $ under averaging obtained by solving \eqref{prob:opt_fourier} with $ \lambda = 0 $ and different energy bounds $ E $.
With sufficient energy, $ \bar{\rho}_{\rm st} $ can accurately approximate the target density. As the energy constraint becomes more stringent, the approximation accuracy deteriorates. Nevertheless, the resulting stationary distribution maintains a structure with three clusters.

\begin{figure}[tb]
  \centering
  \includegraphics[scale=0.35]{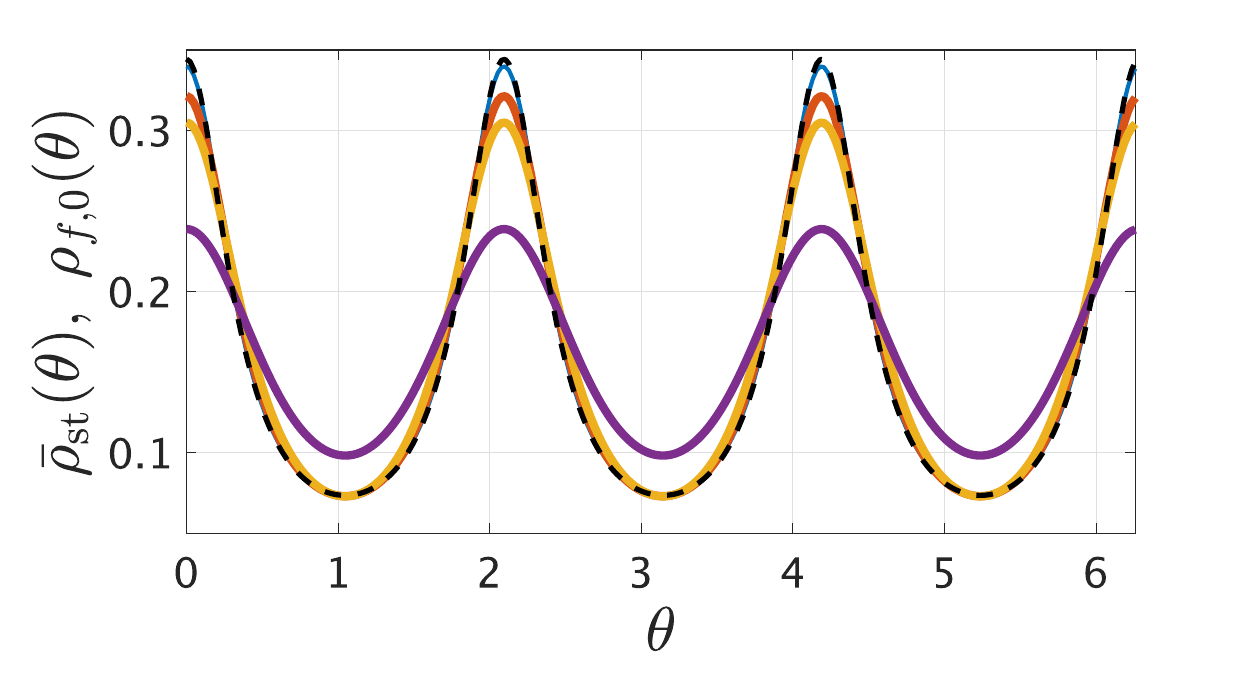}
  \caption{Target distribution $ \rho_{f,0} $ (black, dashed) and the stationary distributions $ \bar{\rho}_{\rm st} $ under averaging obtained by solving \eqref{prob:opt_fourier} with different energy bounds $ E = 0.02 $ (blue), $  0.01 $ (red), $ 0.005 $ (yellow), $  0.001 $ (purple), respectively.}
  \label{fig:rho_st_compare_E}
\end{figure}

\subsection{\hi{Feedback Controller}}
Next, we apply the proposed feedback control law~\eqref{eq:proposed} to the distribution of oscillators.
In what follows, we use $ u_\FF $ obtained by \eqref{prob:opt_fourier} without regularization and the input constraint is set to $ \up(t) \equiv 0.2 $, $ \lw(t) \equiv -0.2 $.
Control of $ \rho $ begins after the surrogate target $ \rho_{\modi} $ has converged to a periodic stationary distribution under the periodic input $ u_\FF $.

\hi{The resulting snapshots of $ \rho $ are illustrated in Fig.~\ref{fig:snapshot}. We see that the distribution of oscillators is steered from $ \rho_0(\theta) = w(\theta,\pi,0.5) $ towards the target $ \rho_f = w(3\theta, 0, 1) $ and is sufficiently close to it at $ t = 500 $.
This is also demonstrated in Fig.~\ref{fig:original_deviation}, where we plot the KL divergence and the $ L^2 $~distance between $ \rho $ and the original target $ \rho_f $.}
\hi{Observe that larger gain $ k $ leads to faster convergence of $ \rho $. 
As mentioned in Subsection~\ref{subsec:convergence_fb}, our method with $ k = 0 $ in Fig.~\ref{fig:original_deviation}\subref{fig:KL_k1}, which corresponds to using only a periodic feedforward input $ u_\FF $, is also guaranteed to steer $ \rho $ to $ \rho_{\modi} $ exactly.}

\begin{figure}[t]
	\begin{minipage}[b]{0.5\linewidth}
		\centering
		\includegraphics[keepaspectratio, scale=0.25]
		{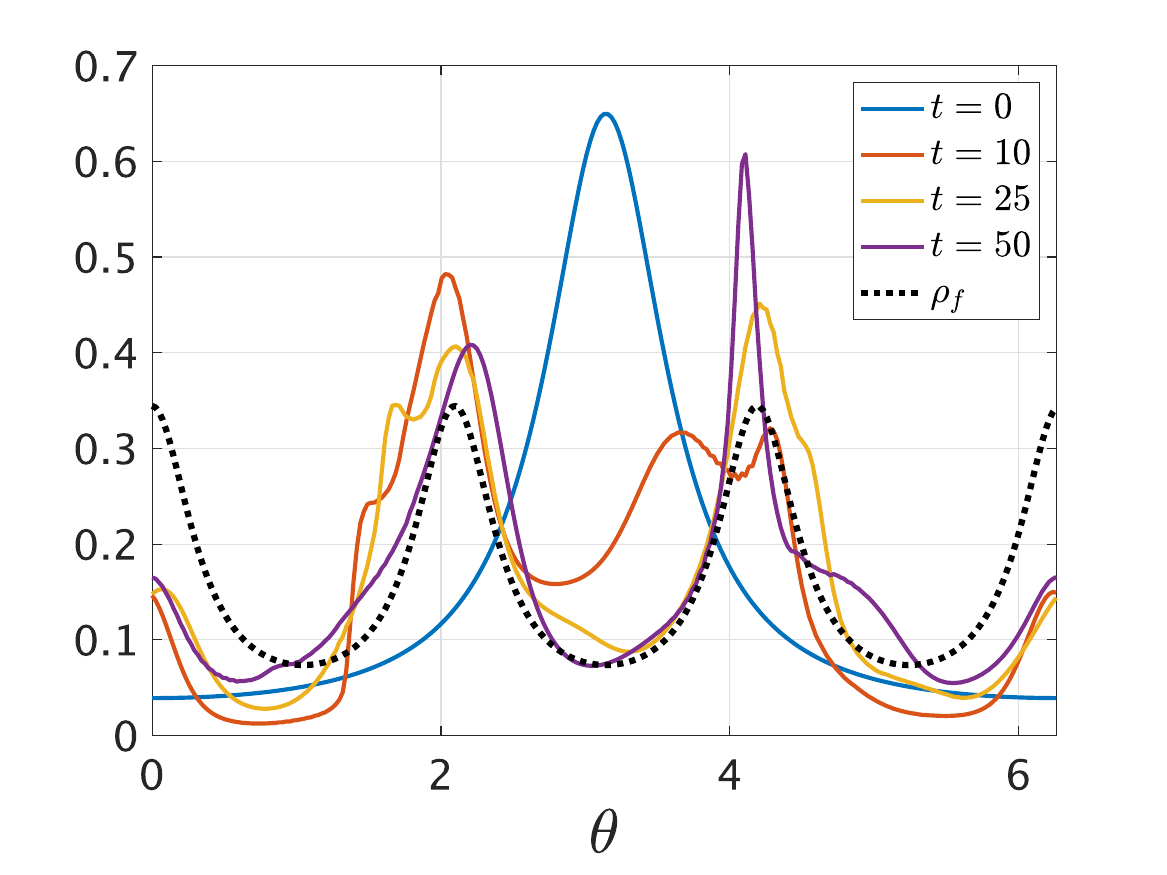}
		\subcaption{}\label{fig:snapshot1}
	\end{minipage}
	\begin{minipage}[b]{0.49\linewidth}
		\centering
		\includegraphics[keepaspectratio, scale=0.25]
		{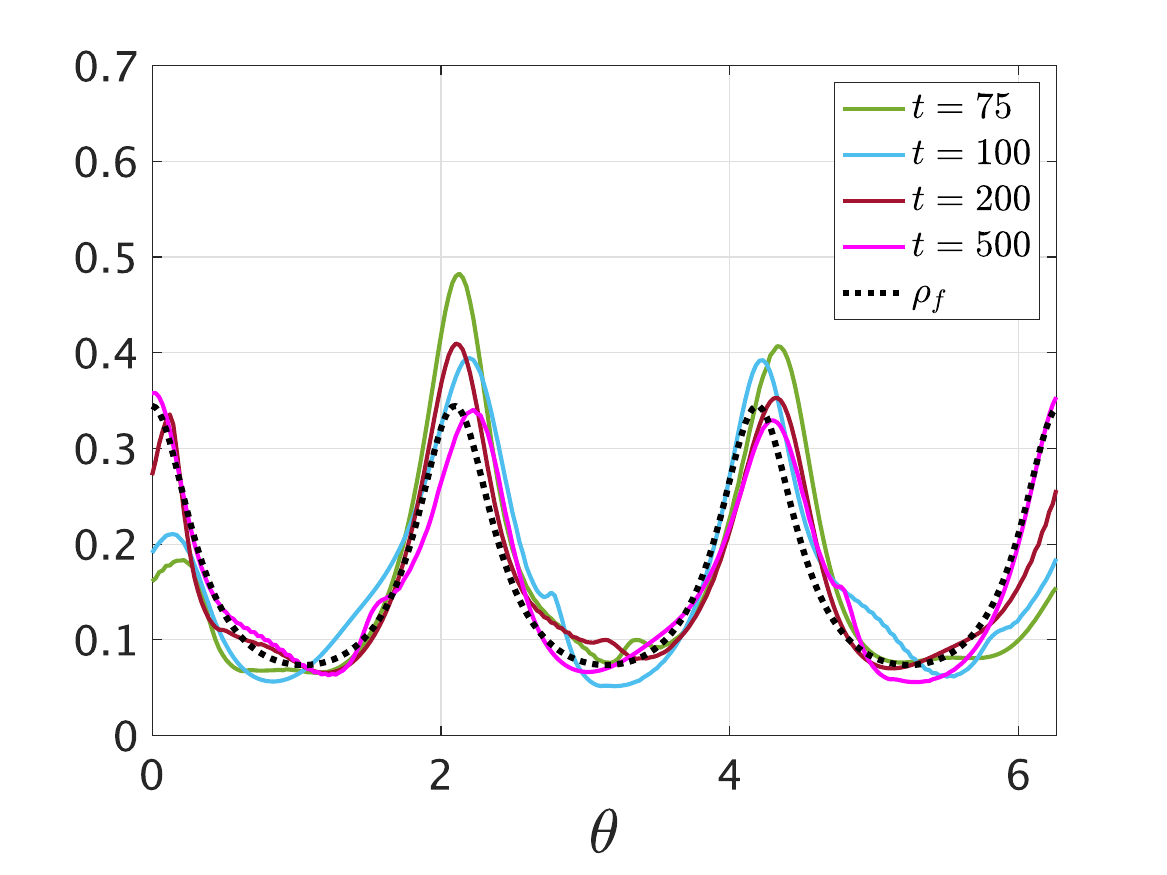}
		\subcaption{}\label{fig:snapshot2}
    \end{minipage}
	\caption{Snapshots of the distribution of oscillators $ \rho $ (solid) and the target $ \rho_f $ (dotted) under the proposed control method and the coordinate transformation $ \theta \rightarrow \theta - \omega t $.}\label{fig:snapshot}
\end{figure}

\begin{figure}[tb]
	\begin{minipage}[b]{0.5\linewidth}
		\centering
		\includegraphics[keepaspectratio, scale=0.24]
		{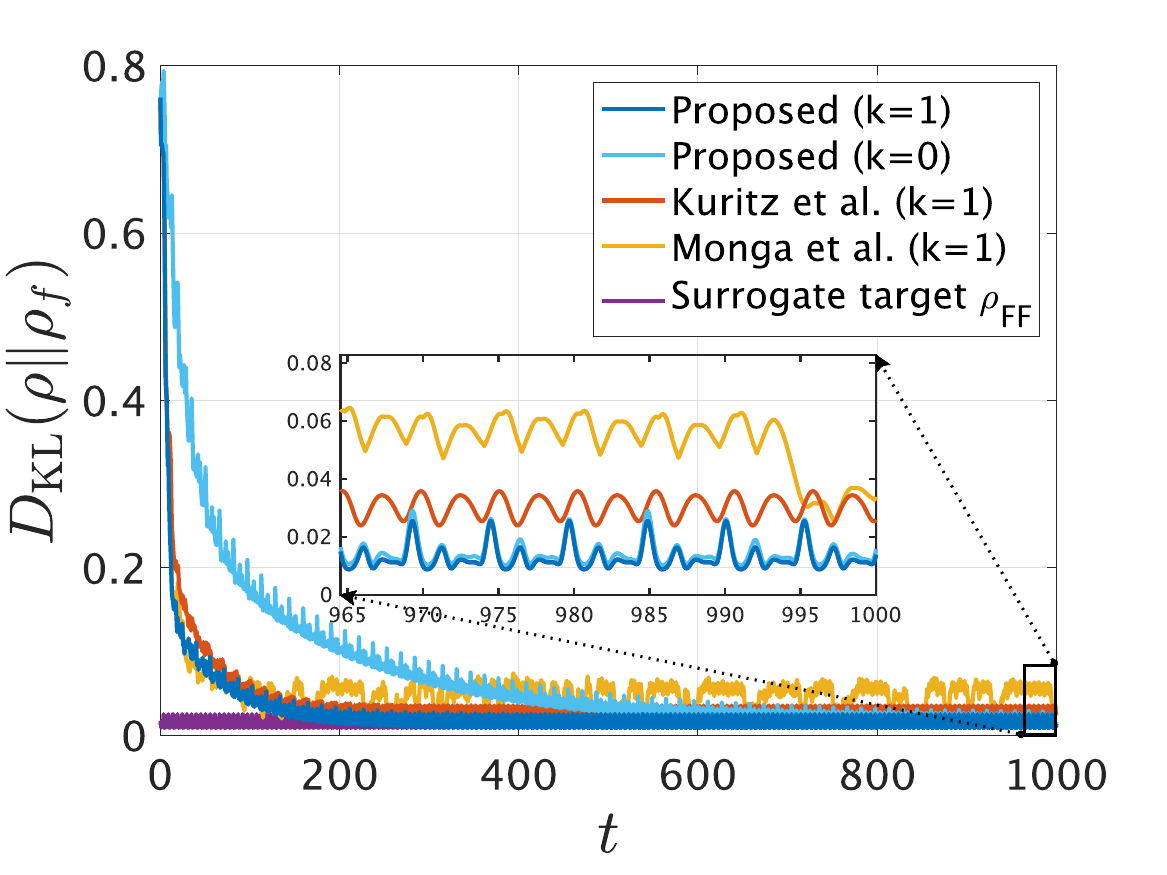}
		\subcaption{}\label{fig:KL_k1}
	\end{minipage}
	\begin{minipage}[b]{0.49\linewidth}
		\centering
		\includegraphics[keepaspectratio, scale=0.24]
		{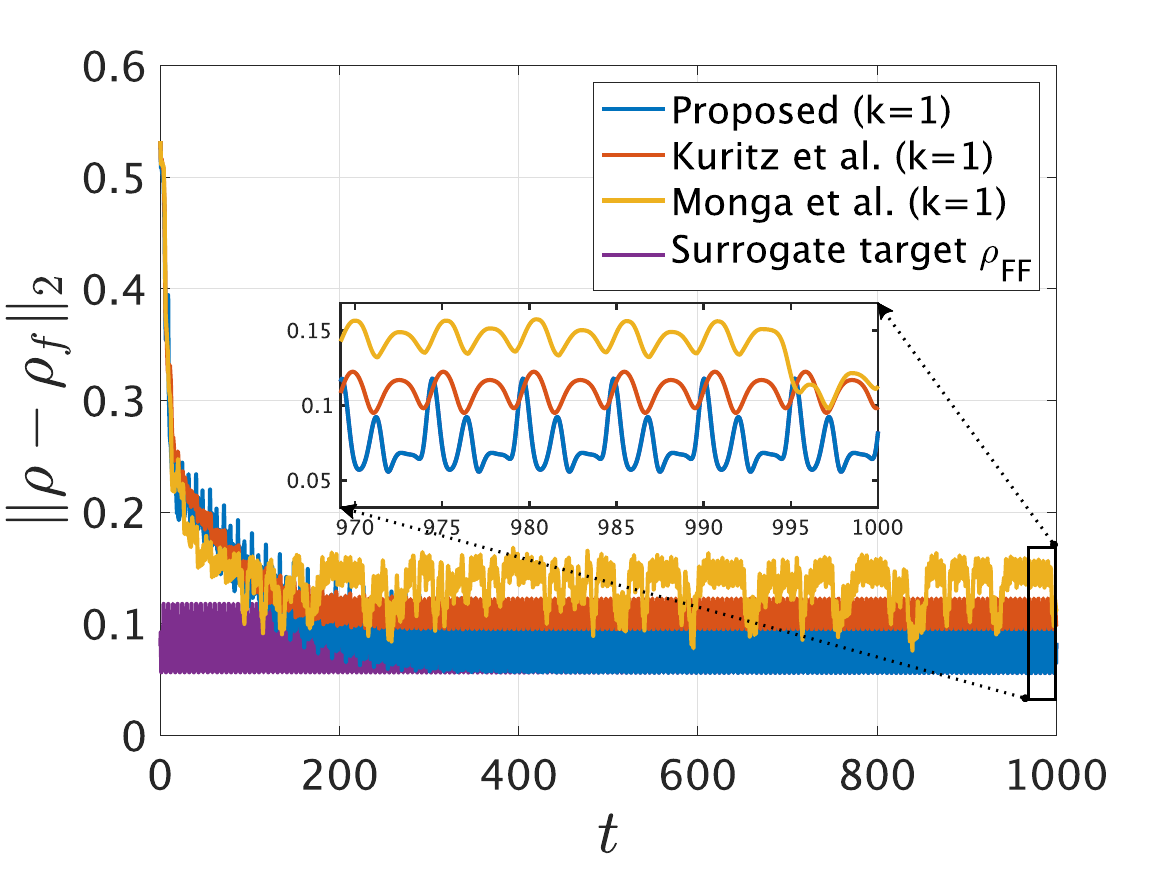}
		\subcaption{}\label{fig:L2_k1}
	\end{minipage}
    	\begin{minipage}[b]{0.5\linewidth}
		\centering
		\includegraphics[keepaspectratio, scale=0.24]
		{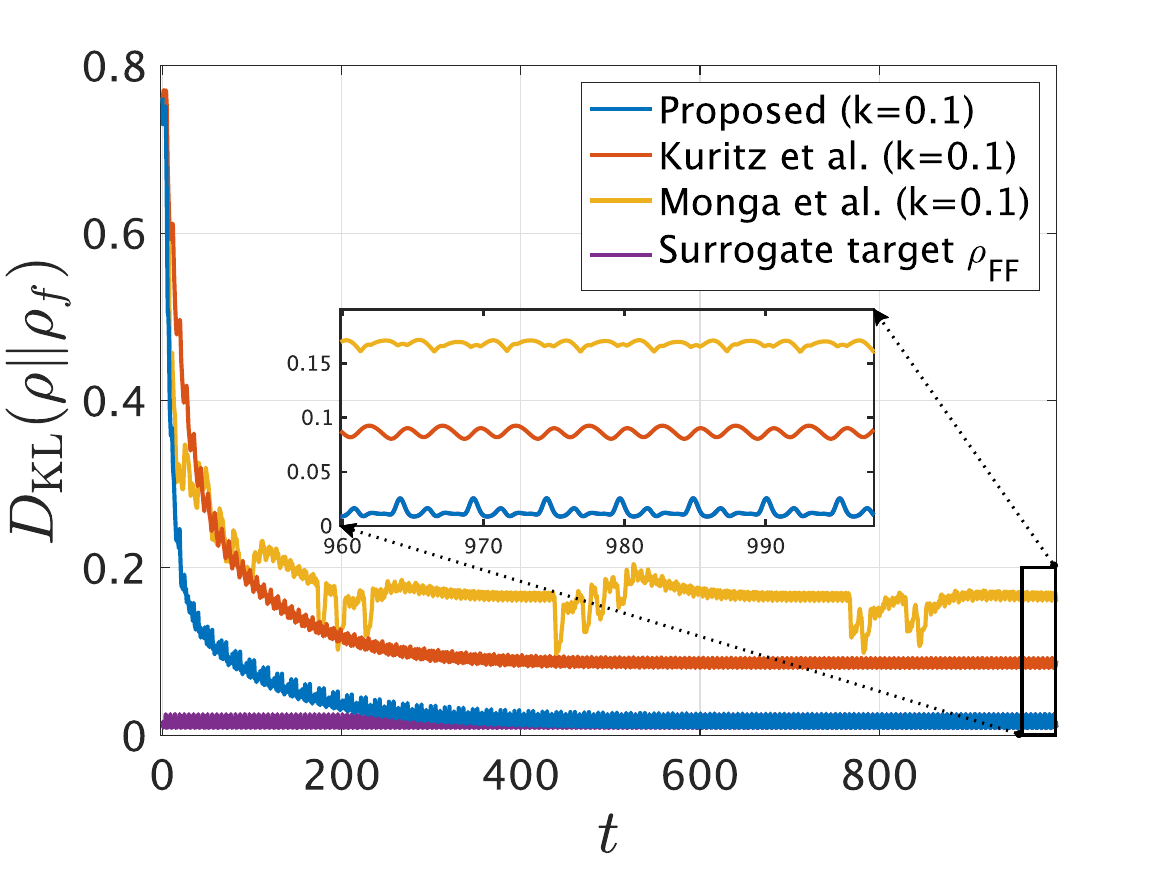}
		\subcaption{}\label{fig:KL_k01}
	\end{minipage}
	\begin{minipage}[b]{0.49\linewidth}
		\centering
		\includegraphics[keepaspectratio, scale=0.24]
		{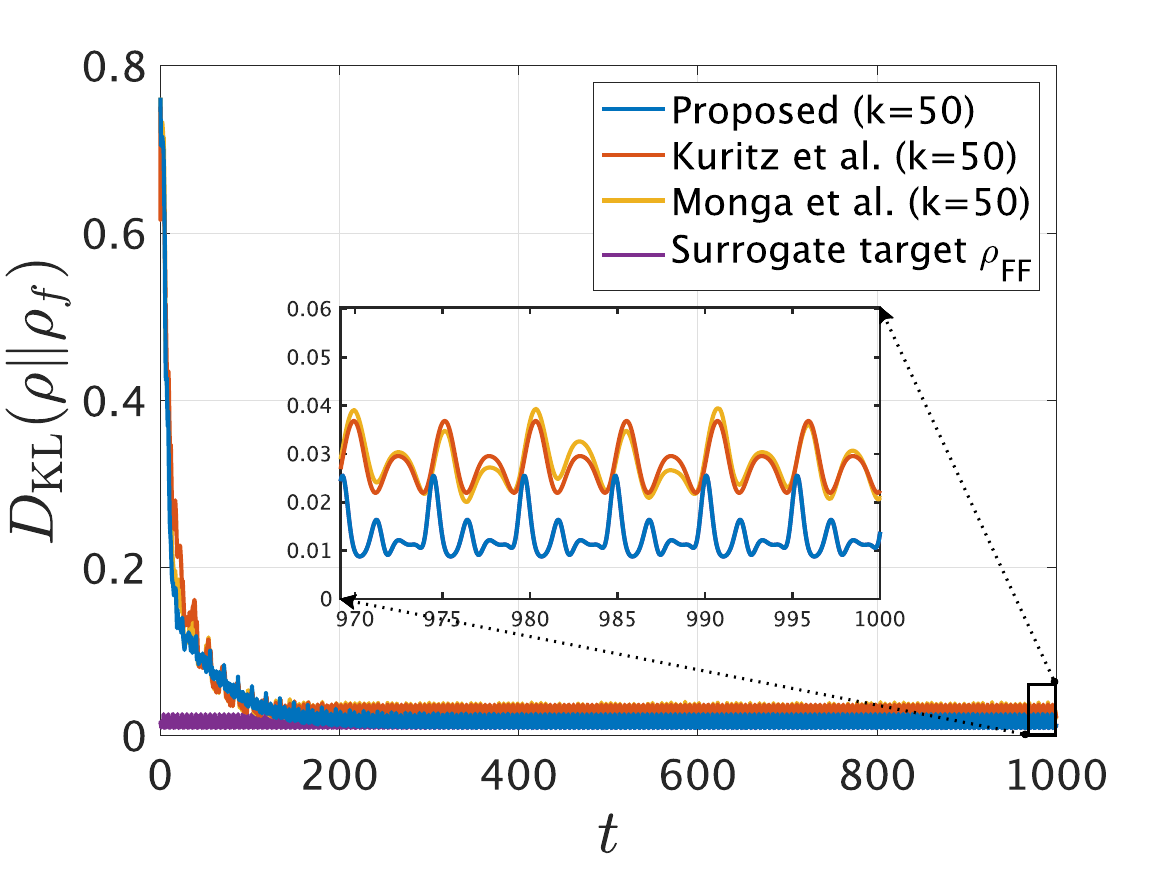}
		\subcaption{}\label{fig:KL_k10}
	\end{minipage}
	\\
	\\
	\begin{minipage}[b]{1\linewidth}
		\centering
		\includegraphics[keepaspectratio, scale=0.35]
		{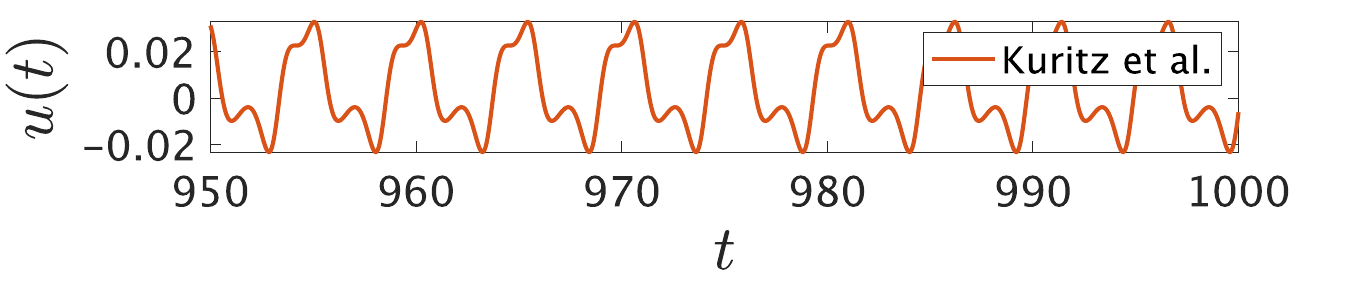}
		\subcaption{}\label{fig:preivous_input}
	\end{minipage}
	\caption{Deviation of the distribution of oscillators $ \rho $ and the surrogate target $ \rho_{\modi} $ from the original target $ \rho_f $ in the KL divergence \subref{fig:KL_k1}, \subref{fig:KL_k01}, \subref{fig:KL_k10} and the $ L^2 $~norm \subref{fig:L2_k1} under the proposed (blue, \eqref{eq:proposed}) and existing methods in \cite{Kuritz} (red, \eqref{eq:previous}) and \cite{Monga} (yellow, \eqref{eq:proposed_cancel}). In \subref{fig:KL_k1} and \subref{fig:L2_k1}, the gain is set to $ k = 1, 0 $ and $ k = 1 $, respectively, and in \subref{fig:KL_k01} and \subref{fig:KL_k10}, $  k = 0.1 $ and $ k = 50 $, respectively. \rev{\subref{fig:preivous_input} Control input by the existing method~\eqref{eq:previous}.}}\label{fig:original_deviation}
\end{figure}

Moreover, in Fig.~\ref{fig:original_deviation}, we compare our method with the existing methods~\cite{Kuritz,Monga} given by \eqref{eq:previous},~\eqref{eq:proposed_cancel} with different gains $ k $.
In all cases, the proposed method steers the distribution of oscillators closest to the target $ \rho_f $ compared to the existing methods.
\rev{To see the reason of this improvement, we plot the control input $ u $ obtained by the existing method~\eqref{eq:previous} in Fig.~\ref{fig:original_deviation}\subref{fig:preivous_input}. It can be seen that the input $ u $ converges to a periodic orbit $ \tilde{u} $ although $ u $ is not made by a periodic input as in our method. 
By Theorem~\ref{thm:proposed}, the distribution $ \rho $ driven by the proposed controller~\eqref{eq:proposed} whose $ u_\FF $ is set to $ \tilde{u} $ converges to the one driven by the existing method~\eqref{eq:previous}.
This implies that the steady-state behavior of $ \rho $ under the existing method can also be achieved by the proposed controller with an appropriate $ u_\FF $.
In our approach, we use the optimal periodic input $ u_\FF $ minimizing the deviation in \eqref{prob:opt_fourier} rather than $ \tilde{u} $. As a result, the proposed method~\eqref{eq:proposed} with the optimal solution $ u_\FF $ transfers $ \rho $ closer to $ \rho_f $ than the existing method~\eqref{eq:previous}.}
This clarifies the advantage of our method.

Note that division by zero arises in the second term of the existing controller \eqref{eq:proposed_cancel} \hi{from \cite{Monga}}. To avoid divergence of input values of \eqref{eq:proposed_cancel}, we set upper and lower bounds of $ u(t) $ as $ -0.2 \le u(t) \le 0.2 $. Nevertheless, as can be seen from Fig.~\ref{fig:original_deviation}, the discrepancy between $ \rho $ and $ \rho_f $ fluctuates significantly under \eqref{eq:proposed_cancel} due to the division by zero.
\revv{Furthermore, while the performance of the existing methods in steering
the distribution towards the target is highly sensitive to the choice of
the control gain \(k\), our proposed approach ensures convergence to the surrogate target $ \rho_{\modi} $ independently of the gain setting.
To make this comparison explicit, Fig.~\ref{fig:gain_sensitivity}
shows the KL divergence averaged over the final period
$
\frac{1}{T_0}
\int_{t_f-T_0}^{t_f}
D_{\rm KL}
\left(
\rho(t,\cdot)
\|
\rho_f
\right)\rmd t
$
for several feedback gains \(k\), where \(t_f=1000\) denotes the final simulation time and \(T_0:=2\pi/\omega\) denotes the natural period of the oscillators. The proposed
method achieves the smallest discrepancy for all tested gains, whereas the existing methods improve as
\(k\) increases but saturate at a larger discrepancy.}

\begin{figure}[t]
    \centering
    \includegraphics[width=0.6\linewidth]{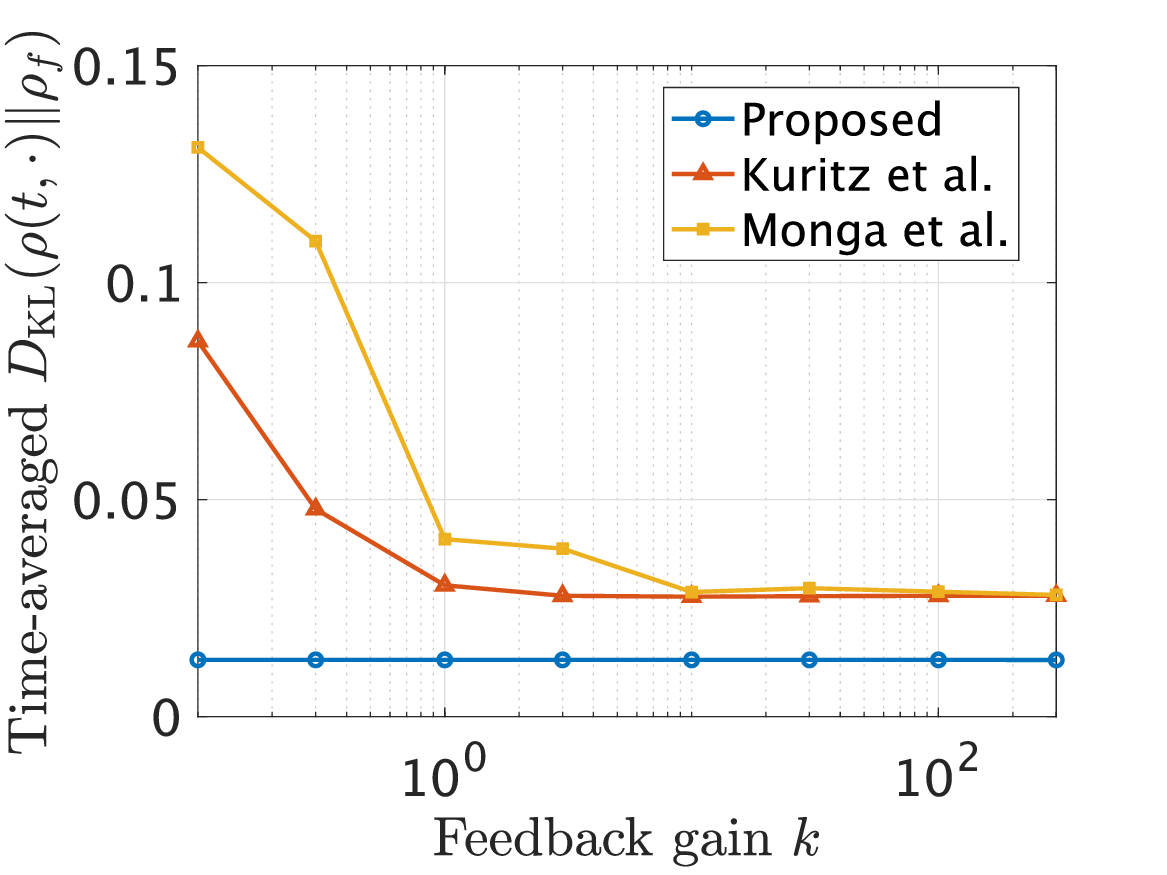}
    \caption{
    Sensitivity of the time-averaged KL divergence $
\frac{1}{T_0}
\int_{t_f-T_0}^{t_f}
D_{\rm KL}
\left(
\rho(t,\cdot)
\|
\rho_f
\right)\rmd t
$ to the feedback
    gain \(k\).
    }
    \label{fig:gain_sensitivity}
\end{figure}

\subsection{\hi{Feedback Controller under Measurement Errors}}
Lastly, we illustrate the usefulness of the proposed controller \eqref{eq:proposed_err} considering measurement errors in Fig.~\ref{fig:deviation_err}.
A perturbed density $ \what{\rho} $ is generated by adding Gaussian noise to the space-discretized density $ \rho $ such that the $ L^2 $~norm of the perturbation is close to $ e $ with high probability. The perturbed density is clipped to be nonnegative and renormalized to integrate to $ 1 $.
As can be seen from Figs.~\ref{fig:deviation_err}\subref{fig:epsi0015},~\subref{fig:epsi0015_flag}, the proposed method initially employs feedback control to accelerate the convergence of $ \rho $ when the distribution $ \rho $ is far from the surrogate target $ \rho_{\modi} $, and then switches to applying only the periodic input $ u_\FF $ after $ \rho $ becomes close enough to $ \rho_{\modi} $.
Consequently, the distribution $ \rho $ converges to the surrogate target $ \rho_{\modi} $ irrespective of $ e $.

\hi{For the existing methods}, Figs.~\ref{fig:deviation_err}\subref{fig:epsi03_kl},~\subref{fig:epsi03_l2} show that as the measurement error increases, the deviation of $ \rho $ from the target becomes large. Although the deviation for the controller \eqref{eq:previous} remains small in the sense of the KL divergence, the $ L^2 $~distance between $ \rho $ and $ \rho_f $ becomes large for large $ e $; see Figs.~\ref{fig:original_deviation}\subref{fig:KL_k1}, \subref{fig:L2_k1} for the cases without measurement errors.
\hi{This can be understood as follows. By definition, the integrand of $ \kl{\rho}{\rho_f} $ is weighted by $ \rho $, and thus, discrepancies between $ \rho $ and $ \rho_f $ in regions where $ \rho $ is small contribute little to $ \kl{\rho}{\rho_f} $. In addition, even in regions where $ \rho_f $ is large, differences between $ \rho $ and $ \rho_f $ lead to only limited increases in $ \kl{\rho}{\rho_f} $. In contrast, for the $L^2$ norm, differences between $ \rho $ and $ \rho_f $ are reflected uniformly. 
}

\begin{figure}[tb]
	\begin{minipage}[b]{0.5\linewidth}
		\centering
		\includegraphics[keepaspectratio, scale=0.24]
		{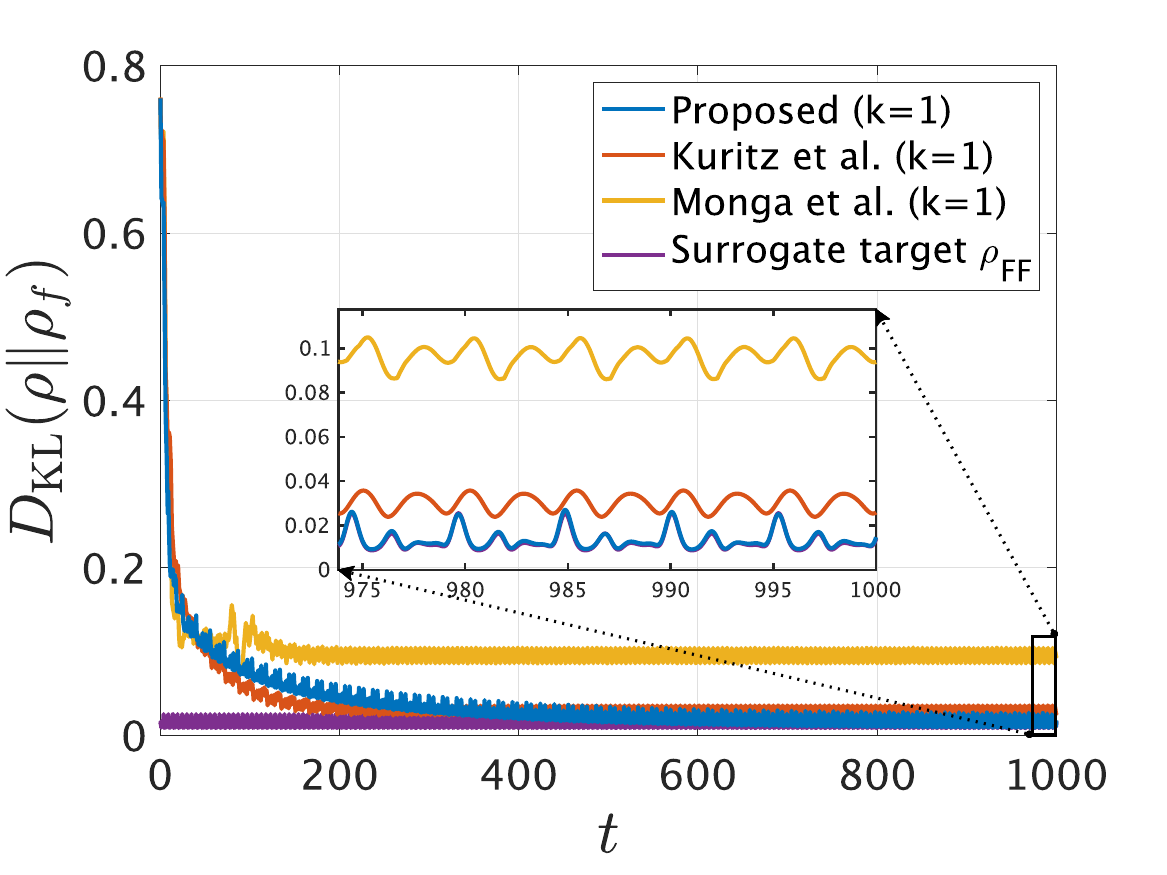}
		\subcaption{}\label{fig:epsi0015}
	\end{minipage}
	\begin{minipage}[b]{0.49\linewidth}
		\centering
		\includegraphics[keepaspectratio, scale=0.24]
		{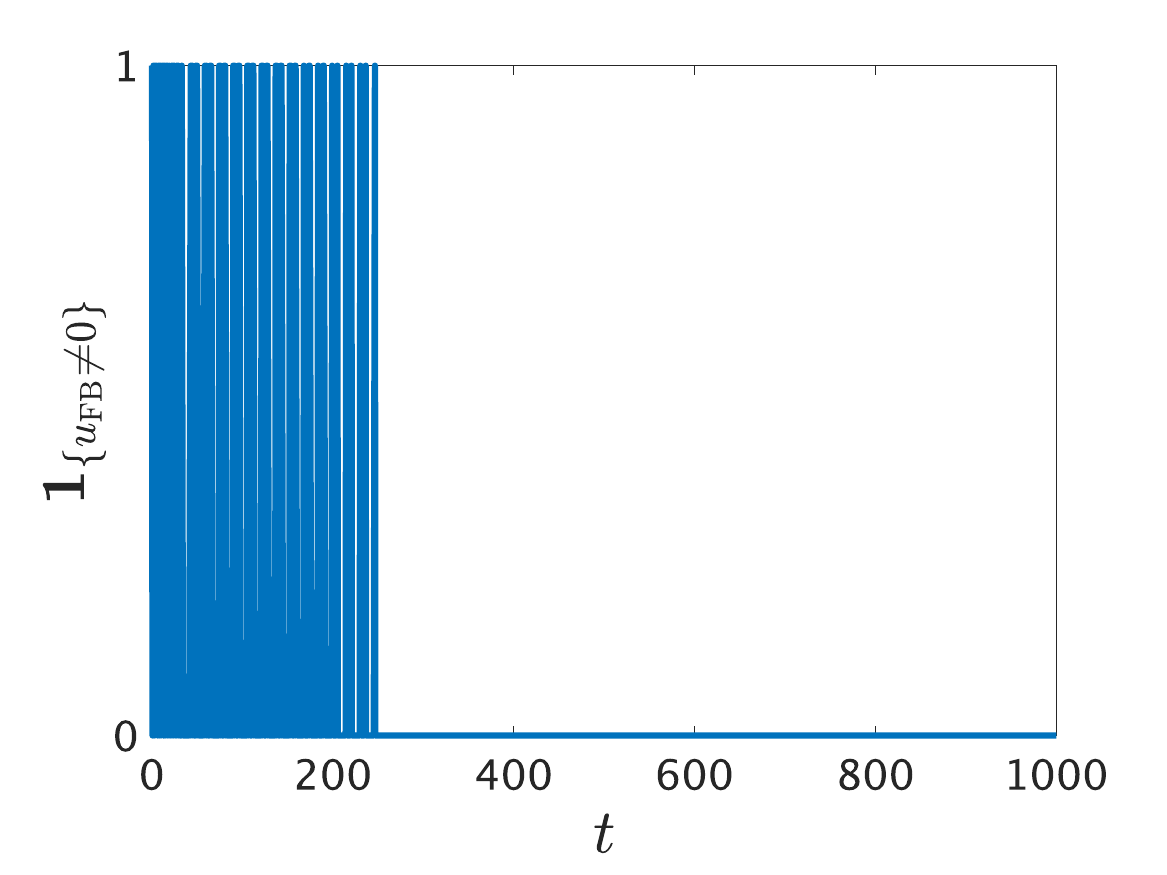}
		\subcaption{}\label{fig:epsi0015_flag}
	\end{minipage}
    	\begin{minipage}[b]{0.5\linewidth}
		\centering
		\includegraphics[keepaspectratio, scale=0.24]
		{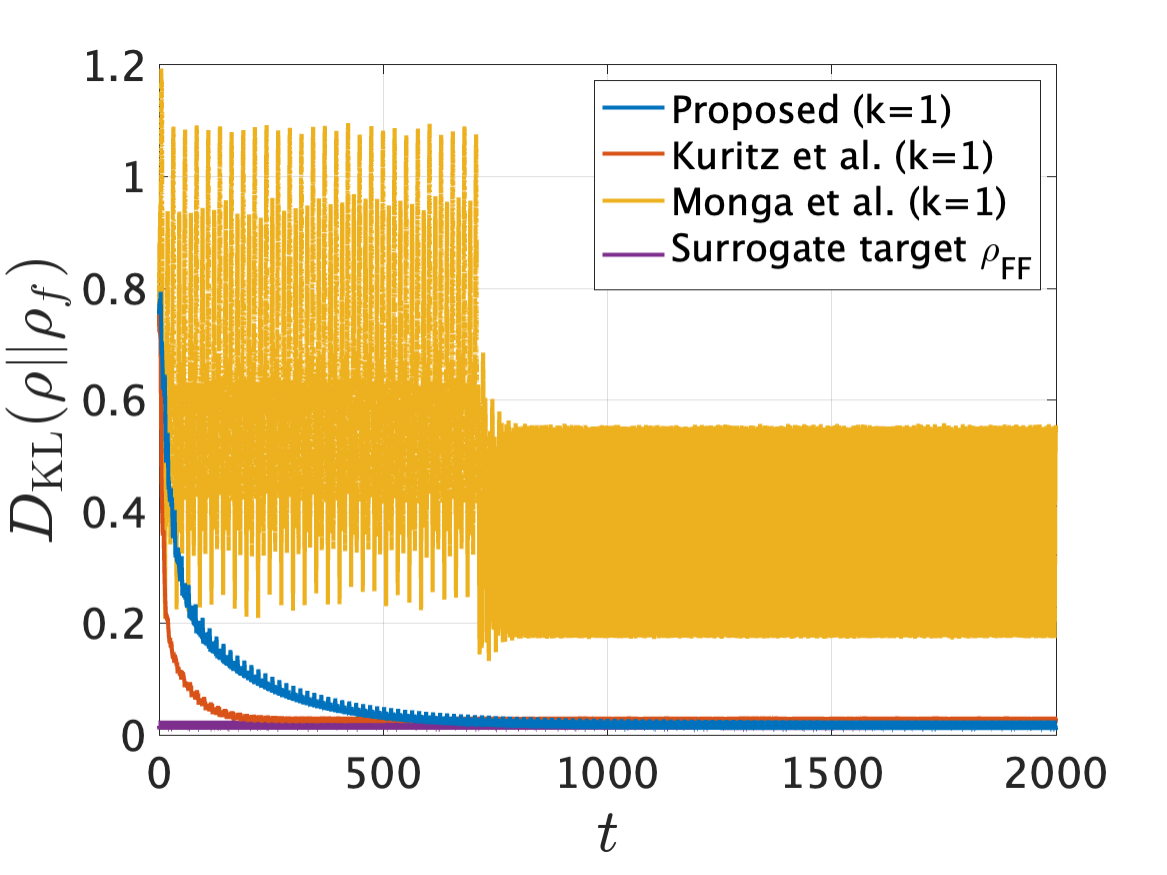}
		\subcaption{}\label{fig:epsi03_kl}
	\end{minipage}
	\begin{minipage}[b]{0.49\linewidth}
		\centering
		\includegraphics[keepaspectratio, scale=0.24]
		{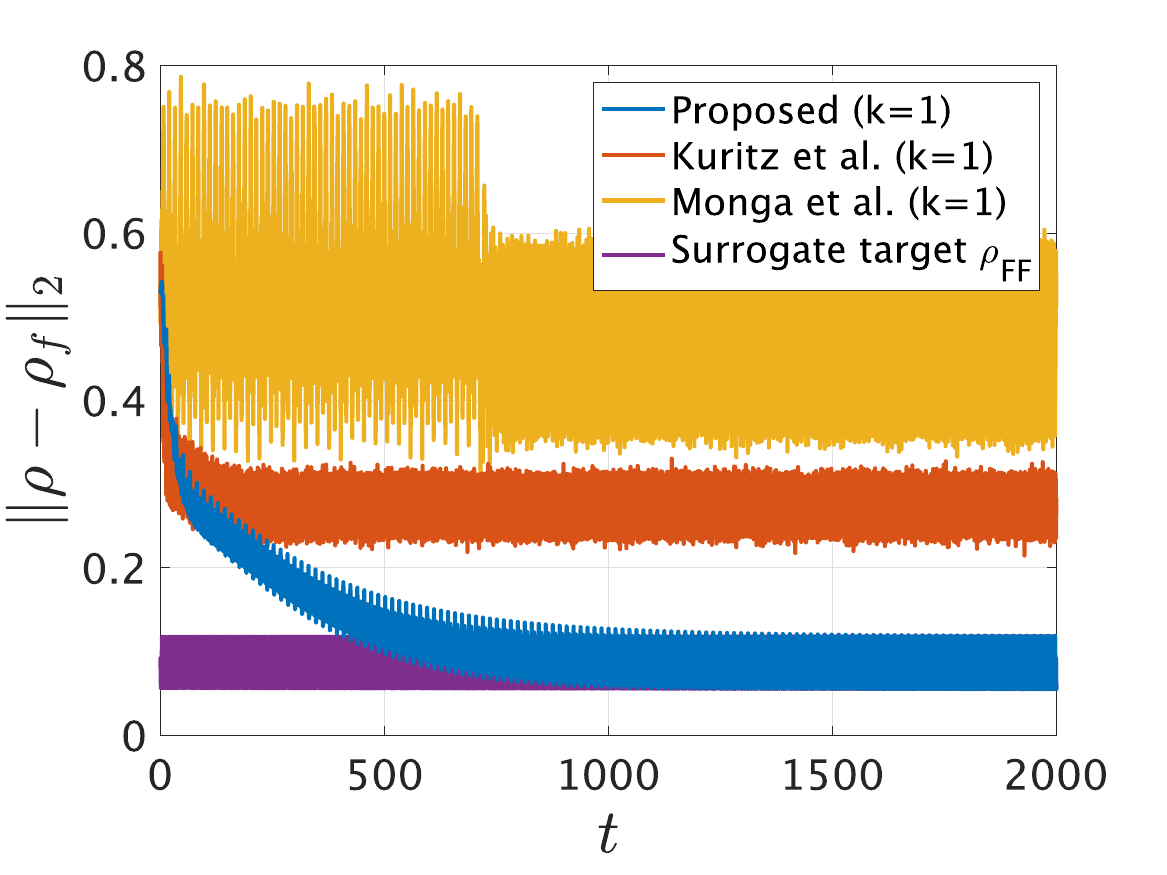}
		\subcaption{}\label{fig:epsi03_l2}
	\end{minipage}
	\caption{Deviation of the distribution of oscillators $ \rho $ and the surrogate target $ \rho_{\modi} $ from the original target $ \rho_f $ in the KL divergence \subref{fig:epsi0015},~\subref{fig:epsi03_kl} and the $ L^2 $~norm \subref{fig:epsi03_l2} under the proposed (blue, \eqref{eq:proposed_err}) and existing methods in \cite{Kuritz} (red, \eqref{eq:previous}) and \cite{Monga} (yellow, \eqref{eq:proposed_cancel}). In \subref{fig:epsi0015}, the measurement error bound is set to $ e = 0.015 $ and in \subref{fig:epsi03_kl} and \subref{fig:epsi03_l2}, the error bound is $  e = 0.3 $. \rev{\subref{fig:epsi0015_flag} Indicator function that equals $ 1 $ if the condition \eqref{eq:decrease_cond} is satisfied, and $ 0 $ otherwise for \subref{fig:epsi0015}.}}\label{fig:deviation_err}
\end{figure}

\section{Conclusion}\label{sec:conclusion}
In this paper, we developed a distribution control method composed of a periodic feedforward input and a population-level feedback controller to transfer the distribution of stochastic oscillators close to a given target distribution.
The periodic input plays a role in keeping the distribution of oscillators close to the target, and the feedback control accelerates the convergence of oscillators.
We showed that the asymptotic reachability of the phase distribution under averaging can be characterized by the Fourier coefficients of the phase sensitivity function and the target distribution. 
Moreover, we established the convergence properties of the proposed method based on the KL divergence.
Numerical simulations demonstrated the effectiveness of the proposed method.

An important future direction is to generalize our idea to coupled oscillators.
For example, distribution control of oscillators coupled through the mean field reduces to analyzing a nonlinear Fokker--Planck equation\hi{\cite{Kawamura2007}}.
\revv{Another important direction is to extend the present framework to heterogeneous oscillator populations, such as the case where the natural frequencies $\omega_i$ are independently sampled at the initial time and remain fixed. This would naturally lead to a distribution control problem on the joint phase-frequency space.}


\appendices

\section{Proof of Lemma~\ref{lem:gibbs}}\label{app:proof_gibbs}

Define $ h_\theta (t) : = Z(\theta + \omega t ) $.
Since $ Z \in C^2 (\calS^1) $, it holds that $ h_\theta (t) \in L^2(0,2\pi/\omega) $. In addition, \revv{its Fourier series converges in the $ L^2 $ sense}, and is given by 
$  
 \sum_{k=-\infty}^\infty z_k \Ee^{\iu k\theta} \Ee^{\iu k \omega t}.
$
Thus, $ \Gamma $ is the inner product of $ L^2 $ functions $ h_\theta $ and $ u_\FF $, and by Parseval's theorem, we have
\begin{align}
	\Gamma (\theta) = \sum_{k=-\infty}^\infty z_k v_{-k} \Ee^{\iu k \theta} . \nonumber
 \end{align}

Moreover, since $ u_\FF\in L^2(0,2\pi/\omega) $ and $ Z \in C^2 (\calS^1) $, we have $ \{v_k\}, \{z_k \} \in \ell^2 (\bbZ) $. By the Cauchy--Schwarz inequality, it holds that
\begin{equation}\label{eq:l1_zv}
	\sum_{k=-\infty}^\infty |z_k v_{-k}| \le \left(\sum_k |z_k|^2\right)^{1/2} \left(\sum_k |v_k|^2\right)^{1/2} < \infty .
\end{equation}
Thus, the Fourier coefficients of $ \Gamma $ converges absolutely and uniformly in $ \theta $.
This enables the following term-by-term integration:
\begin{align}
	V(\theta + 2\pi) - V(\theta) &= - \int_{\theta}^{\theta + 2\pi} \Gamma (\varphi) \rmd \varphi  \nonumber\\
	&= - \int_{\theta}^{\theta + 2\pi} \sum_{k=-\infty}^\infty z_k v_{-k} \Ee^{\iu k \varphi} \rmd \varphi  \nonumber\\
	&= - \hi{2\pi} z_0 v_0. \nonumber
\end{align}
Therefore, it follows from $ v_0 = 0 $ that $ V $ is $ 2\pi $-periodic.

Next, we rewrite $ \bar{\rho}_{\rm st} $ in \eqref{eq:stationary_averaged} as
\begin{align}
	\bar{\rho}_{\rm st}(\theta) &= \frac{1}{C'}\int_{\theta}^{\theta+2\pi}\exp\left(\frac{V(\psi) - V(\theta)}{B^2}\right) \rmd\psi , \label{eq:stationary} 
\end{align}
where we rewrite $ C $ in \eqref{eq:stationary_averaged} as $ C' $.
Then, it follows that
\begin{align}
	\bar{\rho}_{\rm st} (\theta) &= \frac{1}{C'} \exp\left( - \frac{V(\theta)}{B^2} \right) \int_{\theta}^{\theta + 2\pi} \exp\left( \frac{V(\psi)}{B^2} \right) \rmd \psi \nonumber\\
	&= \frac{1}{C} \exp\left( - \frac{V(\theta)}{B^2} \right) , \nonumber
\end{align}
where $ C = C' (\int_{\theta}^{\theta + 2\pi} \exp( V(\psi) /B^2 ) \rmd \psi)^{-1}  $.
This completes the proof.

\section{Proof of Theorem~\ref{thm:exact_reachability}}\label{app:exact_reachability}
First, integrating $ \partial_\theta \log \rho_{f,0} (\theta) = \sum_{k=-\infty}^\infty p_k \Ee^{\iu k \theta} $ over $ \calS^1 $ yields $ p_0 = 0 $.
By Lemma~\ref{lem:gibbs}, we have
\[
	\frac{\Gamma(\theta)}{B^2} = \sum_{k=-\infty}^\infty \frac{z_k v_{-k}}{B^2} \Ee^{\iu k\theta} 
\]
with absolute and uniform convergence on $ \calS^1 $. The definition of $ v_k $, together with the assumption that $ p_k \neq 0 $ means $ z_k \neq 0 $, gives
\[
	\frac{z_k v_{-k}}{B^2} = p_k, \quad \forall k \in \bbZ .
\]
Therefore, $ \Gamma / B^2 $ and $ \partial_\theta \log \rho_{f,0} $ have the same Fourier coefficients, which means that \revv{they are equal in $ L^2 (\calS^1) $. Moreover, since both functions are continuous, they are equal pointwise}:
\[
	\frac{\Gamma(\theta)}{B^2} = \partial_\theta \log \rho_{f,0} (\theta), \quad \forall \theta \in \calS^1 .
\]

By the definition of $ V $ in \eqref{eq:gamma_integral}, we have $ \partial_\theta V(\theta)  = - \Gamma(\theta) $, and thus,
\[
	\partial_\theta \log \bar{\rho}_{\rm st} (\theta) =  - \frac{\partial_\theta V(\theta)}{B^2} = \frac{\Gamma(\theta)}{B^2} = \partial_\theta \log \rho_{f,0} .
\]
Hence, $ \log \bar{\rho}_{\rm st} - \log \rho_{f,0} $ is constant on $ \calS^1 $. Since both $ \bar{\rho}_{\rm st} $ and $ \rho_{f,0} $ are probability densities, this constant must be zero. This concludes that $ \bar{\rho}_{\rm st} = \rho_{f,0} $.

\section{Proof of Theorem~\ref{thm:appro_reachability}}\label{app:appro_reachability}

By Lemma~\ref{lem:gibbs}, we have
\[
	\partial_\theta \log \bar{\rho}_{\rm st} (\theta) = - \frac{\partial_\theta V(\theta)}{B^2} = \frac{\Gamma(\theta)}{B^2} .
\]
In addition, $ \Gamma / B^2 $ has Fourier coefficients $ \{ z_k v_{-k} / B^2 \} $. Since $ \partial_\theta \log \rho_{f,0} $ has Fourier coefficients $ \{p_k\} $, the function
\[ 
h(\theta) := \partial_\theta \log \bar{\rho}_{\rm st} (\theta) - \partial_\theta \log \rho_{f,0} (\theta), \quad \theta \in \calS^1 
\]
has Fourier coefficients $ \{a_k\} $. Moreover, by Lemma~\ref{lem:gibbs}, $ \{z_k v_{-k} / B^2 \} \in \ell^2 (\bbZ) $, and since $ \partial_\theta \log \rho_{f,0} \in L^2 (\calS^1) $, we have $ \{p_k\} \in \ell^2 (\bbZ) $. Thus, $ \{a_k\} \in \ell^2 (\bbZ) $, and Parseval's identity yields
\begin{equation}\label{eq:h_2norm_bound}
	\| h \|_2^2 = 2\pi \sum_{k=-\infty}^\infty |a_k|^2 .
\end{equation}

By \eqref{eq:l1_zv}, it holds that $ \{ z_k v_{-k} / B^2 \} \in \ell^1 (\bbZ) $.
Thus, if $ \{ p_k \} \in \ell^1(\bbZ) $, we have $ \{a_k \} \in \ell^1 (\bbZ) $.
Then, $ \sum_{k \in \bbZ} a_k \Ee^{\iu k \theta} $ converges to $ h $ absolutely and uniformly.
Consequently, we have
\[
	|h(\theta) | \le \sum_{k=-\infty}^\infty |a_k|, \quad \forall \theta \in \calS^1 .
\]
By integrating both sides of the above inequality, we obtain
\begin{align}
	\|h \|_{1} \le 2\pi \sum_{k=-\infty}^\infty \left| a_k \right| . \label{eq:Df_bound}
\end{align}

Note that the continuity of $ \bar{\rho}_{\rm st} $ on $ \calS^1 $ follows from \eqref{eq:stationary_gibbs} and the $ 2\pi $-periodicity of $ V $ under $ v_0 = 0 $; see Appendix~\ref{app:proof_gibbs} for the periodicity.
Then, we can show that there exists $ \theta_0 \in \calS^1 $ such that $ \bar{\rho}_{\rm st} (\theta_0) = \rho_{f,0} (\theta_0) $. Indeed, if $ \bar{\rho}_{\rm st} (\theta) \neq \rho_{f,0} (\theta) $ for all $ \theta \in \calS^1 $, then it must hold that $ \bar{\rho}_{\rm st} (\theta) > \rho_{f,0} (\theta) $ for all $ \theta \in \calS^1 $ or $ \bar{\rho}_{\rm st} (\theta) < \rho_{f,0} (\theta) $ for all $  \theta \in \calS^1 $ by the continuity of $ \bar{\rho}_{\rm st} $ and $ \rho_{f,0} $. However, this contradicts the fact that $ \int_{S^1} \bar{\rho}_{\rm st} \rmd \theta = \int_{S^1}\rho_{f,0} \rmd \theta  =  1 $.

Let $ g := \log \bar{\rho}_{\rm st} - \log \rho_{f,0} $. 
Then, for any $ \theta \in \calS^1 $,
\begin{align}
		|g(\theta)| = \left| \int_{\theta_0}^\theta h (\varphi) \rmd \varphi \right| &\le \|h\|_{\hi{1}} , \label{eq:holder}
	\end{align}
which yields
\begin{align}
	\| g \|_1 \le 2\pi \|h \|_{1}  . \label{eq:poincare}
\end{align}
Finally, applying H\"{o}lder's inequality to the KL divergence and the relative Fisher information, we have
\begin{align}
	\kl{\rho_{f,0}}{\bar{\rho}_{\rm st}} &\le \|\rho_{f,0}\|_{\infty} \|g \|_{1} , \nonumber\\
	\FI{\rho_{f,0}}{\bar{\rho}_{\rm st}} &\le \|\rho_{f,0}\|_{\infty} \| h \|_{2}^2 .\nonumber
\end{align}
Thus, \eqref{eq:fi_bound} and \eqref{eq:kl_bound} follow from \eqref{eq:h_2norm_bound}, \eqref{eq:Df_bound}, and \eqref{eq:poincare}.

\section{Proof of Theorem~\ref{thm:bound_average_original}}\label{app:bound_average_original}

Before the proof, we introduce the following result for the logarithmic Sobolev inequality. 
\begin{lemma}\label{lem:lsi_bounded}
  Let $\rho_2 \in C^1 (\calS^1)$ be a strictly positive density. Then, $\rho_2$ satisfies $\mathrm{LSI}((\underline{\rho}_2 / \bar{\rho}_2)^2)$ where $ \underline{\rho}_2 := \min_{\theta \in \calS^1} \rho_2 (\theta) $ and $ \bar{\rho}_2 := \max_{\theta \in \calS^1} \rho_2(\theta) $.
  \hfill $ \diamondsuit $
\end{lemma}
\begin{proof}
  It is known that the uniform density $ \rho_{\rm uni} = 1/(2\pi) $ satisfies $\mathrm{LSI}(1)$~\cite{Emery,Weissler}. 
	Let $ U $ be a $ C^1 $~function on $ \calS^1 $. Noting that $ U $ is bounded, by \cite[Property~4.6]{Guionnet}, if a density $ \rho_\mu \in C^1(\calS^1) $ satisfies $ {\rm LSI}(\lambda) $ for some $ \lambda > 0 $, then $ \rho_\nu := \Ee^{-U}\rho_\mu / (\int_{\calS^1} \Ee^{-U}\rho_\mu \rmd \theta) $ satisfies $ {\rm LSI} (\lambda\Ee^{-2\osc(U)}) $, where $ \osc(U) := \max_{\theta\in \calS^1} (U(\theta)) - \min_{\theta\in \calS^1} (U(\theta)) $ for $ U\in C^1(\calS^1) $.
	Then, since $\rho_2 = \Ee^{-\log(\rho_2^{-1})}$, we obtain the desired result.
\end{proof}

Now, we are ready to prove Theorem~\ref{thm:bound_average_original}.
Let $ H_\modi (t) := D_{\rm KL} (\bar{\rho}_\modi (t,\cdot) \| \wtilde{\rho}_\modi (t,\cdot)) $.
Then, noting that $ \int \partial_t \bar{\rho}_\modi \rmd \theta = 0 $, we have
\begin{align}
	\frac{\rmd H_\modi (t)}{\rmd t} = \int_{\calS^1} (\partial_t \bar{\rho}_\modi ) \log \frac{\bar{\rho}_\modi}{\wtilde{\rho}_\modi} \rmd \theta - \int_{\calS^1} \frac{\bar{\rho}_\modi}{\wtilde{\rho}_\modi} \partial_t \wtilde{\rho}_\modi \rmd \theta . \nonumber
\end{align}
Let $ R_\modi  :=  \bar{\rho}_\modi / \wtilde{\rho}_\modi $. Then, 
\begin{align}
	\frac{\rmd H_\modi(t)}{\rmd t} &= \int (\bar{\calL} \bar{\rho}_\modi ) \log R_\modi \rmd \theta - \int R_\modi \wtilde{\calL}_{u_\FF} \wtilde{\rho}_\modi \rmd  \theta \nonumber\\
	&= \int \bar{\rho}_\modi \bar{\calL}^\dagger [\log R_\modi] \rmd \theta - \int R_\modi \wtilde{\calL}_{u_\FF} \wtilde{\rho}_\modi \rmd \theta, \label{eq:derivative_H_1}
\end{align}
where $ \bar{\calL}^\dagger := \Gamma \partial_\theta + B^2 \partial_\theta^2  $ is the adjoint operator of $ \bar{\calL} $.
In addition, for the first term of \eqref{eq:derivative_H_1}, we have
\begin{align}
	\bar{\calL}^\dagger [\log R_\modi] &= \frac{\Gamma \partial_\theta R_\modi}{R_\modi} + \frac{B^2 \partial_\theta^2 R_\modi}{R_\modi} - B^2 \left( \frac{\partial_\theta R_\modi}{R_\modi} \right)^2 \nonumber\\
	&= \frac{\bar{\calL}^\dagger R_\modi}{R_\modi} - B^2 \left( \frac{\partial_\theta R_\modi}{R_\modi} \right)^2 . \nonumber
\end{align}
Therefore,
\begin{align}
	\frac{\rmd H_\modi(t)}{\rmd t} &= \int \wtilde{\rho}_\modi \bar{\calL}^\dagger R_\modi \rmd \theta - \int \bar{\rho}_\modi B^2 \left( \frac{\partial_\theta R_\modi}{R_\modi} \right)^2 \rmd \theta \nonumber\\
    &\quad - \int R_\modi \wtilde{\calL}_t \wtilde{\rho}_\modi \rmd \theta \nonumber\\
	&= - B^2  \FI{\bar{\rho}_\modi}{\wtilde{\rho}_\modi} \nonumber\\
	&\quad +  \int R_\modi (\bar{\calL} \wtilde{\rho}_\modi - \wtilde{\calL}_{u_\FF} \wtilde{\rho}_\modi) \rmd \theta . \label{eq:Hd_derivative_1}
\end{align}
For the second term of the above equation, we have
\begin{align}
	&\int R_\modi (\bar{\calL} \wtilde{\rho}_\modi - \wtilde{\calL}_{u_\FF} \wtilde{\rho}_\modi) \rmd \theta = -\int R_\modi \partial_\theta [(\Gamma - A)\wtilde{\rho}_\modi]   \rmd \theta \nonumber\\
	&= \int (\partial_\theta R_\modi) (\Gamma - A)\wtilde{\rho}_\modi \rmd \theta  = \int \frac{\partial_\theta R_\modi}{R_\modi} (\Gamma - A) \bar{\rho}_\modi \rmd \theta   . \label{eq:mismatch_drift}
\end{align}
By applying the Cauchy--Schwarz inequality and Young's inequality to the above equation, we obtain for any $ \delta > 0 $,
\begin{align}
	&\int R_\modi (\bar{\calL} \wtilde{\rho}_\modi - \wtilde{\calL}_{u_\FF} \wtilde{\rho}_\modi) \rmd \theta \nonumber\\
	&\le \FI{\bar{p}_\modi}{\wtilde{\rho}_\modi}^{1/2} \left( \int (\Gamma - A)^2 \bar{p}_\modi \rmd \theta \right)^{1/2} \nonumber\\
	&\le \frac{\delta}{2} \FI{\bar{\rho}_\modi}{\wtilde{\rho}_\modi} + \frac{1}{2\delta} \int (\Gamma - A)^2 \bar{\rho}_\modi \rmd \theta . \nonumber
\end{align}
By substituting the above bound into \eqref{eq:Hd_derivative_1} and noting that $ B^2 = D $, we obtain that for any $ \delta > 0 $ and $ t \ge 0 $,
\begin{align}
	\frac{\rmd H_\modi  (t)}{\rmd t} &\le - \left(D - \delta \right) \FI{\bar{\rho}_\modi}{\wtilde{\rho}_\modi}  \nonumber\\
	&\quad + \frac{1}{2\delta} \int (\Gamma - A(t,\theta))^2 \bar{\rho}_\modi \rmd \theta . \label{eq:H_ineq_1}
\end{align}

Next, we relate $ \FI{\bar{\rho}_\modi}{\wtilde{\rho}_\modi} $ in the right-hand side to $H_\modi (t) = D_{\rm KL} (\bar{\rho}_\modi  \| \wtilde{\rho}_\modi ) $ by the logarithmic Sobolev inequality.
If for some $ \lambda > 0 $, the density $ \wtilde{\rho}_\modi (t,\cdot) $ satisfies $ {\rm LSI}(\lambda) $ for any $ t > 0 $, then by \eqref{eq:H_ineq_1}, we obtain
\begin{align}
    \frac{\rmd H_\modi  (t)}{\rmd t} &\le - 2\lambda \left(D - \delta \right) H_\modi (t) + \frac{\int (\Gamma - A(t,\theta))^2 \bar{\rho}_\modi \rmd \theta}{2\delta} ,  \label{eq:dH_FF_bound}
\end{align}
where $ D - \delta > 0 $ by assumption.
Here, by the same argument as in the proof of \hi{Theorem}~\ref{thm:proposed}, it can be shown that $ \lim_{t \rightarrow \infty} \| \bar{\rho}_\modi (t,\cdot) - \bar{\rho}_{\rm st} \|_1 = 0 $. 
Since any continuous function $ f $ on $ \calS^1 $ is bounded and $ | \int_{\calS^1} f (\theta) \bar{\rho}_\modi (t,\theta) \rmd \theta -  \int_{\calS^1} f (\theta) \bar{\rho}_{\rm st} (\theta) \rmd \theta | \le \| f\|_\infty \| \bar{\rho}_\modi (t,\cdot) - \bar{\rho}_{\rm st} \|_1 $, it holds that 
\[
\lim_{t\rightarrow \infty} \left| \int_{\calS^1} f (\theta) \bar{\rho}_\modi (t,\theta) \rmd \theta -  \int_{\calS^1} f (\theta) \bar{\rho}_{\rm st} (\theta) \rmd \theta \right| = 0  .
\]
This means that for any $ \varepsilon > 0 $, there exists $ t_\varepsilon > 0 $ such that
\begin{equation}\label{eq:ff_st_bound}
    \left| E_\modi (t) - E_{\rm st} (t) \right| \le \varepsilon,  ~~ \forall t \ge t_\varepsilon .
\end{equation}
where
\begin{align*}
    E_\modi (t) &:= \int_{\calS^1} (\Gamma(\theta) - A(t,\theta))^2 \bar{\rho}_\modi (t,\theta) \rmd \theta, \\
    E_{\rm st} (t) &:= \int_{\calS^1} (\Gamma(\theta) - A(t,\theta))^2 \bar{\rho}_{\rm st} (\theta) \rmd \theta .
\end{align*}

By the Gronwall lemma, for any $ \varepsilon > 0 $, we have from \eqref{eq:ff_st_bound}
\begin{align}
    H_\modi (t) &\le H_\modi (t_\varepsilon) \Ee^{-2\lambda (D - \delta ) (t-t_\varepsilon)}  \nonumber\\
    &\quad + \frac{1}{2\delta} \int_{0}^{t-t_\varepsilon} (E_{\rm st} (s) + \varepsilon ) \Ee^{-2\lambda (D - \delta )(t-t_\varepsilon - s)} \rmd s \nonumber \\
    &\le H_\modi (t_\varepsilon) \Ee^{-2\lambda (D - \delta ) (t-t_\varepsilon)} \nonumber\\
    &\quad + \frac{ \sup_{s\in [0,2\pi/\omega)} E_{\rm st} (s) + \varepsilon}{4\lambda \delta(D-\delta)} \left(1 - \Ee^{-2\lambda(D-\delta)(t-t_\varepsilon)} \right) . \nonumber
\end{align}
Therefore, for any $ \varepsilon > 0 $, it holds that
\begin{align}
    \limsup_{t\rightarrow \infty} H_\modi (t) \le \frac{\sup_{s\in [0,2\pi/\omega)} E_{\rm st} (s) + \varepsilon}{4\lambda\delta(D-\delta)} .\label{eq:Hd_bound1}
\end{align}
By the arbitrariness of $ \varepsilon > 0 $, the above bound also holds for $ \varepsilon = 0 $ and is minimized by $ \delta = D/2 $.

By Lemma~\ref{lem:lsi_bounded}, $ \wtilde{\rho}_\modi(t,\cdot) $ satisfies $\mathrm{LSI}((m/M)^2)$ for any $ t > 0 $. 
\revv{Here, the positivity of $ m $ is guaranteed by \cite[Corollary~3.1]{Bogachev2009positive} under the initial positivity $ \rho_{\modi,0} > 0 $, the periodicity and continuity of $ u_{\FF} $, and the non-degeneracy of the diffusion coefficient (Assumption~\ref{ass:nondegenerate}).}
Substituting $ \lambda = (m/M)^2 $ into \eqref{eq:Hd_bound1} and setting $ \delta = D/2 $ complete the proof.

\section{Proof of Theorem~\ref{thm:proposed}}\label{app:proposed}
Let $m := \inf_{t,\theta}\{\rho_{\modi}(t,\theta)\} $ and $ M := \sup_{t,\theta}\{\rho_{\modi}(t,\theta)\}$. Then, by the assumption on $ \rho_{\modi} $, it holds that $ M < \infty $. In addition, $ m > 0 $ is guaranteed under the initial positivity $ \rho_{\modi,0} > 0 $, the periodicity and continuity of $ u_{\FF} $, and the non-degeneracy of the diffusion coefficient~\cite[Corollary~3.1]{Bogachev2009positive}. Therefore, by Lemma~\ref{lem:lsi_bounded}, $ \rho_{\modi}(t,\cdot) $ satisfies $\mathrm{LSI}((m/M)^2)$ for any $ t > 0 $.
In addition, by the assumption $ Z_w^2 > 0 $, the continuity of $ Z_w $, and the compactness of $ \calS^1 $, there exist $ \underline{z} > 0 $ and $ \bar{z} > 0 $ such that $ Z_w^2 (\theta) \in [\underline{z}, \bar{z}] $ for any $ \theta \in \calS^1 $, and thus,
$
		\FI{\rho_z}{\rho_{\modi,z}} \ge (\underline{z} / \bar{z} ) \FI{\rho}{\rho_\modi} 
$.
Moreover, we have $ D_{z,t} \ge D \underline{z} $ for any $ t $.

		Hence, by \eqref{eq:dHdt}, we have
    \begin{align}
        \frac{\rmd H(t)}{\rmd t} \le - \frac{D \underline{z}^2 }{\bar{z}}  \FI{\rho}{\rho_{\modi}} .
    \end{align}
	By the logarithmic Sobolev inequality~\eqref{eq:log_Sobolev} with $ \rho_2 = \rho_{\modi}(t,\cdot) $ and $ \lambda = (m/M)^2 $, it holds that
\begin{equation}
\frac{\rmd H(t)}{\rmd t} \leq -\frac{2D\lambda \underline{z}^2}{\bar{z}}  H(t) ,
\label{eq:H_exp}
\end{equation}
which means the exponentially fast convergence $ H(t) \le \Ee^{-(2D\lambda \underline{z}^2 / \bar{z}) t} H(0) $. \hi{We thus have 1).}

Next, thanks to the \rev{Csisz\'{a}r--Kullback--Pinsker} inequality~\cite{Villani, Gilardoni}:
\begin{equation}
\|\rho-\rho_{\modi}\|_1^2\leq 2H(t),
\end{equation}
we obtain 2).

Lastly, since $ \rho_{\modi} $ satisfies ${\rm LSI}(\lambda)$, Talagrand's inequality
\begin{equation}
  W_2(\rho(t,\cdot), \rho_{\modi}(t,\cdot)) \leq \sqrt{\frac{2}{\lambda}H(t)}
\end{equation}
holds~\cite[Theorem~1]{Otto}, which \hi{implies 3)}.

\section{Derivation of \eqref{eq:rho_rhof}}\label{app:rho_rhof}
Similarly to \eqref{eq:Hd_derivative_1}, we have
\begin{align}
    \frac{\rmd}{\rmd t} \kl{\rho}{\rho_f} = \int R_f (\calL_u - \calL_f) \rho_f \rmd \theta  - D_{z,t} \FI{\rho_z}{\rho_{f,z}} , \label{eq:rho_rhof_derivative}
\end{align}
where $ \rho_{f,z} (t,\theta) := \frac{\rho_f(t,\theta) Z_w^2(\theta)}{\int \rho_f Z_w^2 \rmd \theta} $ and $ R_f(t,\theta) = \frac{\rho(t,\theta)}{\rho_f (t,\theta)} $. The first term can be written as
\begin{align}
    \int R_f (\calL_u - \calL_f) \rho_f \rmd \theta &= - \left( \int R_f \partial_\theta [Z \rho_f] \rmd \theta  \right) u \nonumber\\
    &\quad + D \int R_f \partial_\theta^2 [Z_w^2 \rho_f] \rmd \theta .
\end{align}
Moreover, we have
\begin{align}
    &D \int R_f \partial_\theta^2 [Z_w^2 \rho_f] \rmd \theta \nonumber\\
    &= - D_{z,t} \int \rho_z \left( \partial_\theta \log \frac{\rho_z}{\rho_{f,z}} \right) \partial_\theta \log \rho_{f,z}  \rmd \theta  \nonumber\\
    &\le D_{z,t} \FI{\rho_z}{\rho_{f,z}}^{1/2} \left( \int \rho_z \left( \partial_\theta \log \rho_{f,z} \right)^2\rmd \theta \right)^{1/2} , \label{eq:bound_integral_f}
\end{align}
where we used H\"{o}lder's inequality. 
Lastly, the second component can be bounded as follows:
\begin{align}
    \int \rho_z \left( \partial_\theta \log \rho_{f,z} \right)^2\rmd \theta \le 2\pi \|\rho_z\|_\infty \FI{\frac{1}{2\pi}}{\rho_{f,z}} .
\end{align}
By combining this with \eqref{eq:rho_rhof_derivative}--\eqref{eq:bound_integral_f}, we obtain \eqref{eq:rho_rhof}.

\section{Proof of Proposition~\ref{prop:measurement_err}}\label{app:measurement_err}
We investigate the condition for the nonpositiveness of the first term of \eqref{eq:KL_FF2}, that is,
\begin{align}
	\left( \int \frac{\rho}{\rho_{\modi}} \partial_\theta[Z\rho_{\modi}] \rmd \theta \right)  \left(\int  \frac{\what{\rho}}{\rho_{\modi}}\partial_\theta[Z\rho_{\modi}] \rmd \theta \right)  \ge 0 . \label{eq:positive_cond}
\end{align}
The left-hand side can be written as follows:
\begin{align}
	&\left( \int \frac{\rho - \what{\rho}}{\rho_{\modi}}\partial_\theta [Z\rho_{\modi}] \rmd \theta + \int \frac{\what{\rho}}{\rho_{\modi}} \partial_\theta[Z\rho_{\modi}] \rmd \theta  \right) \nonumber\\
	&\times  \int \frac{\what{\rho}}{\rho_{\modi}} \partial_\theta[Z\rho_{\modi}] \rmd \theta . \nonumber
\end{align}
Therefore, if it holds that
\begin{align}
	\left| \int \frac{\rho - \what{\rho}}{\rho_{\modi}} \partial_\theta [Z\rho_{\modi}] \rmd \theta \right| \le \left|\int \frac{\what{\rho}}{\rho_{\modi}} \partial_\theta[Z\rho_{\modi}] \rmd \theta \right| = \frac{|\what{u}_{\FB}(t)|}{k} , \label{eq:error_sufficient_1}
\end{align}
then \eqref{eq:positive_cond} holds.
Moreover, the left-hand side of \eqref{eq:error_sufficient_1} can be bounded as follows:
\begin{align}
	\left| \int \frac{\rho - \what{\rho}}{\rho_{\modi}} \partial_\theta [Z\rho_{\modi}] \rmd \theta \right| &\le \int \left|\frac{\partial_\theta [Z\rho_{\modi}]}{\rho_{\modi}} \right| |\rho - \what{\rho}| \rmd \theta \nonumber\\
	&\le \left[ \int \left( \frac{\partial_\theta [Z\rho_{\modi}]}{\rho_{\modi}}  \right)^2 \rmd \theta  \right]^{\frac{1}{2}} \| \rho - \what{\rho} \|_{2} . \nonumber
\end{align}
Therefore, if
\begin{align}
	\left[ \int \left( \frac{\partial_\theta [Z\rho_{\modi}]}{\rho_{\modi}}  \right)^2 \rmd \theta  \right]^{1/2} \| \rho - \what{\rho} \|_{2} &\le  \left| \int \frac{\what{\rho}}{\rho_{\modi}} \partial_\theta[Z\rho_{\modi}] \rmd \theta \right| \nonumber\\
    &=\frac{|\what{u}_{\FB}(t)|}{k}, \label{eq:positive_sufficient}
\end{align}
then \eqref{eq:positive_cond} holds. Since the measurement error is bounded as $ \| \rho - \what{\rho} \|_{2} \le e $, we obtain the sufficient condition \eqref{eq:decrease_cond} for \eqref{eq:positive_cond} to hold.
Then, the proof of the convergence of $ \rho $ to $ \rho_\modi $ is \hi{the} same as those of \hi{Proposition}~\ref{prop:rfisher} and \hi{Theorem}~\ref{thm:proposed}.

\section*{References}
\vspace{-\baselineskip}

\bibliographystyle{IEEEtran}
\bibliography{ensemble_oscillator_TAC}

\begin{IEEEbiographynophoto}{Kaito Ito} received the bachelor’s degree in engineering and the \hi{Master’s} and
doctoral degrees in informatics from Kyoto University, Kyoto, Japan, in 2017, 2019, and 2022,
respectively.

From 2022 to 2024, he was an Assistant Professor with the Department of Computer Science, Tokyo Institute of Technology, Yokohama,
Japan. He is currently an Assistant Professor with the Department of Information Physics
and Computing, The University of Tokyo, Tokyo,
Japan. His research interests include stochastic control, optimization,
and machine learning.

\end{IEEEbiographynophoto}

\begin{IEEEbiographynophoto}{Haruhiro Kume} received the bachelor's and \hi{Master's degrees}
in engineering from
Tokyo Institute of Technology, Yokohama, Japan, in 2021 and 2023, respectively.
His research interests include control of a population of oscillators.
\end{IEEEbiographynophoto}

\begin{IEEEbiographynophoto}{Hideaki Ishii} received the
M.Eng. degree from Kyoto University in 1998,
and the Ph.D. degree from the University of
Toronto in 2002. He was a Postdoctoral Research Associate at the University of Illinois at
Urbana-Champaign in 2001--2004, and a Research Associate at The University of Tokyo in
2004--2007. He was an Associate Professor and
then a Professor at the Department of Computer
Science, Tokyo Institute of Technology in 2007--
2024. Currently, he is a Professor at the Department of Information Physics and Computing at The University of Tokyo
since 2024. He was a Humboldt Research Fellow at the University of
Stuttgart in 2014--2015. He has also held visiting positions at CNR-IEIIT
at the Politecnico di Torino, the Technical University of Berlin, and the
City University of Hong Kong. His research interests include networked
control systems, multiagent systems, distributed algorithms, and cybersecurity of control systems.

Dr. Ishii has served as an Associate Editor for Automatica, the
IEEE Control Systems Letters, the IEEE Transactions on Automatic
Control, the IEEE Transactions on Control of Network Systems, and the
Mathematics of Control, Signals, and Systems. He was a Vice President
for the IEEE Control Systems Society (CSS) in 2022--2023 and an
elected member of the IEEE CSS Board of Governors in 2014--2016.
He was the Chair of the IFAC Coordinating Committee on Systems and
Signals in 2017--2023 and the Chair of the IFAC Technical Committee
on Networked Systems for 2011--2017. He served as the IPC Chair for
the IFAC World Congress 2023 held in Yokohama, Japan. He received
the IEEE Control Systems Magazine Outstanding Paper Award in 2015.
Dr. Ishii is an IEEE Fellow.

\end{IEEEbiographynophoto}

\end{document}